\documentclass[11pt]{article}
\usepackage{fullpage}
\usepackage{amsmath,amsfonts,amssymb,amsthm}
\usepackage{xspace}
\usepackage{tikz}
\usepackage{paralist}
\usepackage{framed}
\usepackage{xparse}
\usepackage{hyperref}
\usepackage{subcaption}
\usepackage{ifthen}
\usetikzlibrary{decorations.pathreplacing, decorations.pathmorphing}

\newtheorem{theorem}{Theorem}
\newtheorem{lemma}[theorem]{Lemma}

\newtheoremstyle{restate}{}{}{\itshape}{}{\bfseries}{~(restated).}{.5em}{\thmnote{#3}}
\theoremstyle{restate}
\newtheorem*{restate}{}

\theoremstyle{definition}
\newtheorem{definition}[theorem]{Definition}
\theoremstyle{remark}

\newcommand{\E}{\mathop{\mathbb{E}}}
\newcommand{\poly}{\mathrm{poly}}
\allowdisplaybreaks

\newcommand{\DKLver}[2]{D_{\mathrm{KL}}\left(\ooalign{$\genfrac{}{}{1.6pt}0{#1}{#2}$\cr$\color{white}\genfrac{}{}{.8pt}0{\phantom{#1}}{\phantom{#2}}$}\right)}
\newcommand{\dist}{\mathrm{dist}}

\newcommand{\cA}{\mathcal{A}}
\newcommand{\cB}{\mathcal{B}}
\newcommand{\cD}{\mathcal{D}}
\newcommand{\cS}{\mathcal{S}}
\newcommand{\cP}{\mathcal{P}}
\newcommand{\tV}{\tilde{V}}
\newcommand{\tv}{\tilde{v}}
\newcommand{\good}{\textrm{good}}
\newcommand{\urdec}{\mathsf{UR_{dec}^{\subset}}}
\newcommand{\ur}{\mathsf{UR^{\subset}}}
\newcommand{\conn}{\textrm{conn}}
\newcommand{\blk}{\textrm{blk}}
\newcommand{\sk}{\textrm{sk}}
\newcommand{\con}{\textsf{connectivity}}
\newcommand{\bcast}{\mathsf{BCAST}}
\newcommand{\ucast}{\mathsf{UCAST}}

\newsavebox{\mybox}
\NewDocumentEnvironment{code}{mO{Algorithm}O{()}O{#1}}
	{\begin{center}
	\begin{lrbox}{\mybox}\footnotesize%
	\begin{minipage}{5.5in}
	\ifcsname c@pctr#1\endcsname\else\newcounter{pctr#1}\textbf{#2 $\mathtt{#4}#3$}:\fi
	\setlength{\topsep}{0pt}
	\begin{compactenum}
	\setcounter{enumi}{\value{pctr#1}}
	}{\setcounter{pctr#1}{\value{enumi}}
	\end{compactenum}
	\end{minipage}
	\end{lrbox}\fbox{\usebox{\mybox}}
	\end{center}}

\title{Tight Distributed Sketching Lower Bound for Connectivity}
\author{Huacheng Yu\thanks{Department of Computer Science, Princeton University. \texttt{yuhch123@gmail.com}}}
\date{}

\begin{document}

\setcounter{page}{0}
\maketitle
\thispagestyle{empty}
\begin{abstract}
	In this paper, we study the distributed sketching complexity of \con{}.
	In distributed graph sketching, an $n$-node graph $G$ is distributed to $n$ players such that each player sees the neighborhood of one vertex.
	The players then simultaneously send one message to the \emph{referee}, who must compute some function of $G$ with high probability.
	For \con{}, the referee must output whether $G$ is connected.
	The goal is to minimize the message lengths.
	Such sketching schemes are equivalent to one-round protocols in the broadcast congested clique model.

	We prove that the expected average message length must be at least $\Omega(\log^3 n)$ bits, if the error probability is at most $1/4$.
	It matches the upper bound obtained by the AGM sketch~\cite{AGM12}, which even allows the referee to output a spanning forest of $G$ with probability $1-1/\poly\, n$.
	Our lower bound strengthens the previous $\Omega(\log^3 n)$ lower bound for \emph{spanning forest} computation~\cite{NY19}.
	Hence, it implies that \con{}, a decision problem, is as hard as its ``search'' version in this model.
\end{abstract}
\newpage

\section{Introduction}
In distributed graph sketching, an $n$-node graph $G$, where the nodes are labeled with integers from $1$ to $n$, is distributed to $n$ players such that the $i$-th player can see the labels of neighbors of node $i$.
Then each player, based on this information, sends a short message (called the \emph{sketch}) to a special player, called the \emph{referee}.
Finally, the referee, who does not have direct access to the graph, must compute some function of $G$. 
We usually assume that all players (including the referee) have access to shared random bits.
The goal is to minimize the message lengths.
In this paper, we study \con{} in this model, i.e., the referee must decide whether $G$ is connected.

For problems where the referee must compute a function that depends on the \emph{global} structure of $G$, it may seem that the players have no way to figure out what information is more important from only the \emph{local} structures.
For instance for \con{}, each player only sees the set of edges incident to a node, and they cannot distinguish which edges are more crucial in connecting $G$ (e.g., bridges).
Hence, it may seem that they must tell the referee a large amount of information so that the ``important'' information is included in the message with high probability.
Surprisingly, this intuition is wrong.
Ahn, Guha and McGregor~\cite{AGM12} showed that it is possible to ``sketch'' each neighborhood using only $O(\log^3 n)$ bits, such that the referee is still able to decide if $G$ is connected with high probability.

Roughly speaking, in their algorithm, each player computes ``hashes'' of its neighborhood such that the hash values allow the referee to recover one neighbor of this vertex.
Moreover, these hashes are ``mergeable'', in the sense that by combining the hashes of a set of vertices, the referee is able to recover one edge \emph{from this set to the rest of $G$}.
Therefore, by repeatedly finding outgoing edges from each connected component, and merging the connected components and their hash values, the referee will be able to decide if $G$ is connected.
We present a more detailed summary in Section~\ref{sec_agm}.

In a previous work of Nelson and Yu~\cite{NY19}, it was shown that if the referee has to output the entire spanning forest with constant probability, then the sketch size has to be $\Omega(\log^3 n)$ bits.
Computing a spanning forest could, in principle, be a much harder task, as the output has $\Theta(n\log n)$ bits, implying a trivial lower bound of $\Omega(\log n)$ bits.
On the other hand, \con{} only requires the referee to learn \emph{one bit} about $G$.
In this paper, we strengthen the previous lower bound, and show that the players still have to send $\Omega(\log^3 n)$ bits on average in order for the referee to learn this one-bit.

\newcommand{\thmmaincont}{For any (randomized) distributed sketching scheme that allows the referee to decide if $G$ is connected with probability at least $3/4$, the average sketch size of all players must be at least $\Omega(\log^3 n)$ bits in expectation.}
\begin{theorem}\label{thm_main}
	\thmmaincont
\end{theorem}

\subsection{Related work}

\paragraph{Distributed computing.}
Distributed graph sketching is related to the \emph{broadcast congested clique} model ($\bcast(b)$) in distributed computing, where \con{} has attracted significant attention lately~\cite{BKMNRST15,MT16,JN17,JN18,PP19}.
In $\bcast(b)$, an input $G$ is distributed to $n$ players, such that the $i$-th player sees the neighborhood of vertex $i$.
An algorithm in this model proceeds in rounds.
In each round, each player simultaneously \emph{broadcasts} one message of length $b$ to all other players, and performs (free) local computation.
After broadcasting (and receiving from every other player) $r$ messages, the players must figure out the output.
The goal is to minimize the number of rounds $r$.
Therefore, distributed graph sketching asks what is the smallest $b$ such that the problem admits a one-round protocol in $\bcast(b)$.

Montealegre and Todinca~\cite{MT16} showed that one can solve \con{} \emph{deterministically} in $O(r)$ rounds in $\bcast(n^{1/r}\log n)$.
Jurdzi\'{n}ski and Nowicki~\cite{JN17} improved that round complexity to $O(\log n/\log\log n)$ for $\bcast(\log n)$.
Their algorithm is also deterministic.
For randomized algorithms, the AGM sketch~\cite{AGM12} solves the problem with only one round in $\bcast(\log^3 n)$.
Pai and Pemmaraju~\cite{PP19} proved an $\Omega(\frac{1}{b}\log n)$ round lower bound in $\bcast(b)$ for deterministic \con{} algorithms.
They also showed the same lower bound for random algorithms that compute the \emph{connected components} of $G$.
To the best of our knowledge, Theorem~\ref{thm_main} is the first non-trivial lower bound for \con{} in $\bcast(b)$ for $b=\omega(\log n)$, even for deterministic algorithms (it implies that if $b=o(\log^3 n)$, then we need at least two rounds).

A related model in distributed computing is the \emph{unicast congested clique} model ($\ucast(b)$), where each player is allowed to send \emph{possibly different} messages to other players in each round.
It turns out that the $\ucast(b)$ model is much more powerful than $\bcast(b)$, and \con{} algorithms with significantly lower round complexity exist~\cite{LPPP05,HPPSS15,GP16,JN18b}.
The best known algorithms use $O(\log\log n)$ rounds deterministically~\cite{LPPP05}, and use $O(1)$ rounds if we allow randomization~\cite{JN18b}.

\paragraph{Dynamic streams.}
The best known distributed sketching scheme~\cite{AGM12} uses \emph{linear sketches}.
If we view the input graph $G$ as a $\binom{n}{2}$-dimensional binary vector $X$, the concatenation of all messages turns out to be a matrix-vector product $AX$, for $A\in\mathbb{Z}^{n\log^2n\times \binom{n}{2}}$ a matrix determined by the shared random bits.
The product vector $AX$ determines if $G$ is connected with high probability.

It also gives an $O(n\log^3 n)$-bit \emph{streaming algorithm} for \con{}.
That is, we wish to maintain a dynamic graph $G$ under edge insertions and deletions using as little memory as possible, such that after all updates, the algorithm is able to decide whether the \emph{final} graph is connected with high probability.
An algorithm can easily maintain the product $AX$ under edge insertions and deletions given $A$.
Moreover, the connectivity of $G$ can be determined from the final $AX$.
Therefore, by maintaining $AX$ (and storing a succinct representation of $A$), this problem can be solved using $O(n\log^3 n)$ bits of memory, where the extra $\log n$ factor is due to the bit complexity of each coordinate of $AX$.

The best known space lower bound for \con{} in this setting is $\Omega(n\log n)$ bits due to Sun and Woodruff~\cite{SW15}.
Their lower bound also holds if we only insert edges to $G$.
It was shown in~\cite{NY19} that if the algorithm has to output a \emph{spanning forest} with constant probability, then it must use at least $\Omega(n\log^3 n)$ bits of space.
Unfortunately, it is not clear whether our new technique can be extended to streaming.

\subsection{AGM sketch}\label{sec_agm}

To better motivate our hard instance and the lower bound argument, we present a summary of the $O(\log^3 n)$ algorithm in this subsection.
The algorithm begins by giving every \emph{possible} undirected edge a unique label, e.g., by concatenating the labels of the two endpoints with the smaller label first.
The basic hash (linear sketch) for each player is simply the XOR of the labels of all incident edges.
Each basic hash takes $O(\log n)$ bits, and it allows one to recover the incident edge if the degree of that vertex happen to be one.

Next, we subsample the edges and compute the basic hashes of each sample.
Specifically, the subsampling creates samples of $O(\log n)$ levels.
In level $i$, we sample each edge with probability $2^{-i}$, and each player computes the basic hash of all \emph{surviving} edges.
Overall, the hashes from all $O(\log n)$ levels have $O(\log^2 n)$ bits.
Now regardless of the degree of the vertex, with constant probability there exists some level where exactly one incident edge survives the sampling.
Hence, the hash at that level recovers this edge.
It turns out that there is a separate structure that has the same size and detects if each level has exactly one surviving edge with error probability $1/\poly\, n$.
Therefore, this $O(\log^2 n)$-bit hash allows one to recover one incident edge with constant probability.

The most important feature of this hash is its mergeability.
That is, if we take two vertices $x$ and $y$, and compute the level-wise XOR of their hashes for each of the $O(\log n)$ levels, then for each level, we obtain simply the XOR of all surviving edges that are incident to either of them, \emph{with the exception that} the edge between $x$ and $y$, if exists, is XORed twice and thus canceled.\footnote{Here, it is important that the players have access to shared randomness, as it allows them to have the same outcome in sampling.}
In general, if we take the level-wise XOR of hashes of a set of vertices, we obtain for each level, the XOR of all surviving \emph{outgoing} edges from this set.
In particular, with constant probability, there exists one level with exactly one outgoing edge, which allows one to recover it.

Finally, each player computes $O(\log n)$ independent such $O(\log^2 n)$-bit hashes, and sends them to the referee.
Therefore, the message lengths are $O(\log^3 n)$ bits.
The referee uses the first $O(\log^2 n)$-bit hashes to compute one outgoing edge from each vertex.
Then the referee merges the hashes of vertices that are already connected, and uses the second hashes to compute one outgoing edge from each connected component, and so on.
It succeeds on each component with constant probability each time.
Therefore, by repeating the above procedure $O(\log n)$ times, the referee recovers the connected components of $G$ with high probability.

\medskip

In summary, the first $\log n$ factor in space is needed to encode the label of a vertex.
The second $\log n$ is used to ``guess'' approximately the number of outgoing edges.
The last $\log n$ factor serves two purposes: The algorithm has $O(\log n)$ rounds, and each round uses fresh randomness; the $O(\log^2 n)$-bit hash only succeeds with constant probability on each connected component, $O(\log n)$ instances are used to ensure that all components succeed.
An $\Omega(\log^3 n)$ lower bound argument and the corresponding hard instance must simultaneously capture the above three factors.

\subsection{Organization}
We present overviews of the previous and new lower bounds in Section~\ref{sec_overview}.
In Section~\ref{sec_urdec_lb}, we prove a lower bound for a communication problem, called $\urdec$.
Finally, we prove the main theorem in Section~\ref{sec_red}, by reducing from $\urdec$.
\section{Overview}\label{sec_overview}
In this section, we summarize the previous lower bound for spanning forest computation~\cite{NY19}, and present an overview of our lower bound proof.
Both lower bound proofs are based on reductions from variants of the communication problem \emph{universal relation}.

\begin{definition}[$\ur$]
In the $\ur$ problem, there are two players Alice and Bob.
Alice receives a set $S\subseteq [U]$ and Bob receives a \emph{proper} subset $T\subset S$ as their inputs.
Then Alice sends one message to Bob, and Bob must find \emph{some} element in $S\setminus T$ with probability at least $1-\delta$.
\end{definition}

The original version (the search version) of universal relation (called $\ur$) is used in the previous spanning forest lower bound.
By applying the above subsampling trick and sending the hashes, this task can be accomplished with $O(\log(1/\delta)\log^2 U)$ bits of communication~\cite{FIS08}.
It turns out that this is optimal as long as $\delta>2^{-U^{0.99}}$~\cite{KNPWWY17}.
The previous spanning forest lower bound is based on a reduction from $\ur$ for $U=n^{\Theta(1)}$ and $\delta=n^{-\Theta(1)}$, in which case, the optimal bound is $\Theta(\log^3 n)$.

In order to prove a lower bound for \con{}, which is a decision problem, we first define and prove a lower bound for a \emph{decision} version of universal relation, called $\urdec$, then we reduce \con{} from it.

\begin{definition}[$\urdec$]
	In the $\urdec$ problem, Alice receives a set $S\subseteq [U]$, Bob receives a \emph{proper} subset $T\subset S$ and a partition $(P_1,P_2)$ of $[U]\setminus T$.
	It is promised that either $S\setminus T\subseteq P_1$, or $S\setminus T\subseteq P_2$.
	Alice sends one message to Bob, and Bob must decide which part contains $S\setminus T$ with probability $1-\delta$.
\end{definition}
Clearly, this is an easier problem than $\ur$, since if Bob could recover any element in $S\setminus T$, then by checking if this element is in $P_1$ or $P_2$, he would be able to decide if $S\setminus T\subseteq P_1$ or $P_2$.
In Section~\ref{sec_urdec_lb}, we prove in fact, the decision version is as hard as the search version.

It turns out that the reductions from $\ur$ and $\urdec$ to spanning forest and \con{} respectively have similar main structures.
On the other hand, the previous $\ur$ lower bound strategy~\cite{KNPWWY17} completely fails on the decision problem $\urdec$.
Hence, the \emph{main technical contribution} of this paper is the $\urdec$ lower bound proof.
In the subsections below, we first overview the reductions from the universal relation problems, and then present a summary of their communication lower bounds.

\subsection{Previous reduction from $\ur$}
Now let us see what is the connection between $\ur$ and distributed sketching for spanning forest.
Fix a vertex $v$.
The player at $v$ sees its neighborhood $S$, and sends a message $M_v$ to the referee.
Suppose the referee figures out that there is a subset $T$ of neighbors of $v$, which have $v$ as their only neighbor.
Then, the only way for $\{v\}\cup T$ to connect to the rest of the graph is through some edge from $v$ to $S\setminus T$.
In the other words, in order to output any spanning forest, the referee must find an element in $S\setminus T$.
Since $v$ does not know $T$ and the referee does not know $S$, intuitively the communication between them must at least ``solve'' $\ur$.

However, this argument does not directly give us a proof.
The main issue is that in distributed sketching, every edge is \emph{shared between} two players.
In particular, the other ``endpoint'' in $S\setminus T$ also knows this edge.
Therefore, any vertex $u$ who has $v$ as its neighbor can simply tell the referee this fact, and the referee learns an element (vertex) in $S\setminus T$ from the message of that vertex.
To resolve this issue, we put a ``large number'' of independent ``$v$'' and a ``small number'' of ``other endpoints'' in the graph, so that the total amount of information revealed by the other ``endpoints'' becomes negligible.
More specifically (see also Figure~\ref{fig_hard_a}), we random permute the labels, and pick a set of vertices $V^m$ to be all potential $v$.
For each $v^m_i\in V^m$, we independently construct a $\ur$ instance $(S_i, T_i)$ such that all vertices in $T_i$ have $v^m_i$ as their only neighbor ($V^l_i$ in Figure~\ref{fig_hard_a}) and all vertices in $S_i\setminus T_i$ are contained in a \emph{much smaller} set $V^r$.
Each $v^m_i$ sees a randomly labeled set of neighbors $S_i$, and as in $\ur$, the player does not know $T_i$ .
Moreover, since $|V^r|\ll |V^m|$, the total information that can be revealed by the other ``endpoint'' of $S\setminus T$ is at most $|V^r|\cdot \poly\log n\ll |V^m|$ (otherwise \emph{some} vertex in $V^r$ must send a very long message).
For an \emph{average} $v^m_i$, this information is negligible.
By a standard information theoretic argument, we can show that for an {average} $v^m_i$, even if the referee does not receive messages from $V^r$, he can still find a neighbor of $v^m_i$ in $V^r$ with high probability.
It then implies that if there is a spanning forest protocol, then one can solve $\ur$ with the same communication and approximately the same error probability.

The final graph $G$ consists of $\sqrt{n}$ \emph{independent} blocks of size $\sqrt{n}$, where each block is constructed as above.
If the referee can find a spanning forest with constant probability, i.e., find a spanning tree in all blocks, then one can show that for one block, the referee must be able to find its spanning tree with probability $1-O(1/\sqrt{n})$.
Hence, by applying the above argument on one block, we may reduce the problem from $\ur$ with error probability $\approx 1/\sqrt{n}$ and $U=|V^r|=n^{\Theta(1)}$.
As we mentioned above, there is an $\Omega(\log^3 n)$ $\ur$ lower bound under this setting of parameters, implying an $\Omega(\log^3 n)$ lower bound for spanning forest.

\begin{figure}
	\begin{center}
		\begin{subfigure}{0.3\linewidth}
			\centering
			\begin{tikzpicture}[vtx/.style={draw, circle, inner sep=0, minimum size=3pt}]
				\foreach \i in {0,1,2,3,4}
					\node[vtx] (A\i) at (0pt, \i*25 pt) {};
				\foreach \i in {3}
					\foreach \j in {0,1,2}
						\node[vtx] (B\i\j) at (-40pt, \i*25+\j*5-5 pt) {};
				\foreach \i in {1,4}
					\foreach \j in {1}
						\node[vtx] (B\i\j) at (-40pt, \i*25+\j*5-5 pt) {};
				\foreach \i in {0,2}
					\foreach \j in {0,2}
						\node[vtx] (B\i\j) at (-40pt, \i*25+\j*5-5 pt) {};
				\foreach \i in {0,1,2}
					\node[vtx] (C\i) at (40pt, \i*35+15 pt) {};
				\draw [rounded corners=5pt] (-8pt, -10pt) rectangle (8pt, 110pt);
				\draw [rounded corners=3pt] (-45pt, 39pt) rectangle (-35pt, 61pt);
				\draw [rounded corners=4pt] (33pt, 7pt) rectangle (47pt, 93pt);

				\node at (-40pt, -20pt) {\scriptsize $V^l$};
				\node at (0, -20pt) {\scriptsize $V^m$};
				\node at (40pt, -20pt) {\scriptsize $V^r$};
				\foreach \i in {1,2,3,4,5}
					\node at (-55pt, \i*25-25 pt) {\scriptsize $V^l_\i$};

				\draw (B41) -- (A4) -- (C1);
				\draw (A4) -- (C0);

				\draw (B22) -- (A2) -- (C0);
				\draw (B20) -- (A2);
			\end{tikzpicture}
		\caption{}\label{fig_hard_a}
		\end{subfigure}
		\quad
		\begin{subfigure}{0.3\linewidth}
			\centering
			\begin{tikzpicture}[vtx/.style={draw, circle, inner sep=0, minimum size=3pt}]
				\foreach \i in {0,1,2,3,4}
					\node[vtx] (A\i) at (0pt, \i*25 pt) {};
				\foreach \i in {3}
					\foreach \j in {0,1,2}
						\node[vtx] (B\i\j) at (-40pt, \i*25+\j*5-5 pt) {};
				\foreach \i in {1,4}
					\foreach \j in {1}
						\node[vtx] (B\i\j) at (-40pt, \i*25+\j*5-5 pt) {};
				\foreach \i in {0,2}
					\foreach \j in {0,2}
						\node[vtx] (B\i\j) at (-40pt, \i*25+\j*5-5 pt) {};
				\foreach \i in {0,1}
					\foreach \j in {0,1}
						\node[vtx] (C\i\j) at (40pt, \i*60+\j*30+5 pt) {};
				\draw [rounded corners=5pt] (-8pt, -10pt) rectangle (8pt, 110pt);
				\draw [rounded corners=3pt] (-45pt, 39pt) rectangle (-35pt, 61pt);
				\draw [rounded corners=4pt] (33pt, -5pt) rectangle (47pt, 45pt);
				\draw [rounded corners=4pt] (33pt, 55pt) rectangle (47pt, 105pt);

				\node at (-40pt, -20pt) {\scriptsize $V^l$};
				\node at (0, -20pt) {\scriptsize $V^m$};
				\node at (40pt, -20pt) {\scriptsize $V^r$};
				\foreach \i in {1,2,3,4,5}
					\node at (-55pt, \i*25-25 pt) {\scriptsize $V^l_\i$};
				\node at (57pt, 20pt) {\scriptsize $V^r_1$};
				\node at (57pt, 80pt) {\scriptsize $V^r_2$};

				\draw (B41) -- (A4) -- (C10);
				\draw (A4) -- (C11);

				\draw (B22) -- (A2) -- (C00);
				\draw (B20) -- (A2);
			\end{tikzpicture}
		\caption{}\label{fig_hard_b}
		\end{subfigure}
	\end{center}
	\caption{$T_i$ is $V^l_i$, and $S_i\setminus T_i$ is contained in $V^r$. }\label{fig_hard}
\end{figure}
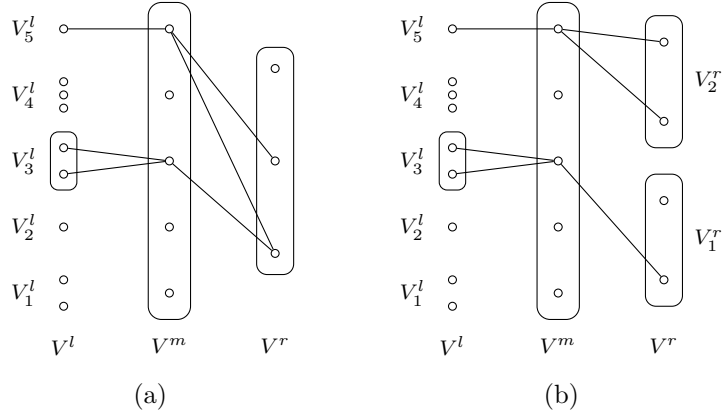

\subsection{Overview of our reduction}
To make a reduction from $\urdec$ to \con{}, we begin by modifying the construction for each block (see Figure~\ref{fig_hard_b}).
We split the set $V^r$ into two sets $V^r_1$ and $V^r_2$.
Then for each vertex $v^m_i\in V^m$, we ensure that its neighbors in $V^r$ are either all in $V^r_1$ or all in $V^r_2$.
As before, the neighborhood of $v^m_i$ corresponds to a set $S_i$, its neighbors in $V^l_i$ corresponds to its subset $T_i$.
Now, let $P_1=V^r_1$ and $P_2=V^r_2$, then $S_i\setminus T_i$ is either a subset of $P_1$ or a subset of $P_2$.
Based on which case it is, $v^m_i$ is either only connected to $V^r_1$, or only connected to $V^r_2$.
What remains is to combine the blocks into a graph $G$ that forces the players to solve $\urdec$ instances with high probability (see Figure~\ref{fig_hard_2}).

For each block, we construct two \emph{identical} copies of a subgraph as above, and denote their vertex sets by ${}^{+}V^l,{}^{+}V^m, {}^{+}V^r$ and ${}^{-}V^l,{}^{-}V^m, {}^{-}V^r$ respectively.
Then, we add four special vertices $s_1,s_2,t_1,t_2$ to the block.
We connect $s_1$ to a \emph{random} ${}^{+}v^m_i$, and connect $s_2$ to its copy ${}^{-}v^m_i$.
Then we connect $t_1$ to all vertices in ${}^{+}V^r_1$ and ${}^{-}V^r_2$, and connect $t_2$ to all vertices in ${}^{-}V^r_1$ and ${}^{+}V^r_2$.
Now, the block has two connected components.
It is easy to verify that each vertex is either in the same connected component with $t_1$ or $t_2$, but $t_1$ and $t_2$ are in different components.
Moreover, $s_1$ and $s_2$ are also in different components.
This is because if ${}^{+}v^m_i$, the only neighbor of $s_1$, has a neighbor in ${}^{+}V^r_1$, then ${}^{-}v^m_i$ has a neighbor in ${}^{-}V^r_1$, in which case, $s_1$ and $t_1$ are in the same connected component, and $s_2$ and $t_2$ are in the same connected component, and vice versa.
Thus, we construct a block such that either
\begin{compactenum}[(i)]
	\item $s_1$ and $t_1$ are in the same component, $s_2$ and $t_2$ are in the same component; or
	\item $s_1$ and $t_2$ are in the same component, $s_2$ and $t_1$ are in the same component.
\end{compactenum}
Deciding which is the case requires the referee to determine for this random vertex $v^m_i$, whether its neighbors are in $V^r_1$ or $V^r_2$, i.e., ``solving'' the $\urdec$ instance embedded at $v^m_i$.

We independently construct $\sqrt{n}$ such blocks, and add an edge between $t_1$ [resp. $t_2$] of block $i$ and $s_1$ [resp. $s_2$] of block $i+1$ (where block $\sqrt{n}+1$ is block $1$).
This graph is connected if and only if there is an \emph{odd} number of blocks where case (ii) above happens.
That is, deciding if the whole graph is connected is equivalent to computing the XOR of $\sqrt{n}$ bits, one for each block.
It turns out that in the distributed sketching model, if the referee computes the XOR with $3/4$ probability, then for most blocks, the referee can decide which case this block is in with probability $1-O(1/\sqrt{n})$.
By the same argument as before, it allows us to reduce \con{} from $\urdec$ with $U=n^{\Theta(1)}$ and $\delta=1/n^{\Theta(1)}$.
The formal proof can be found in Section~\ref{sec_red}.

\begin{figure}
	\begin{center}
		\begin{subfigure}{0.9\linewidth}
			\centering
			\begin{tikzpicture}[vtx/.style={draw, circle, inner sep=0, minimum size=3pt}]
				\foreach \c in {-1,1} {
					\foreach \i in {0,1,2,3,4}
						\node[vtx] (A\c\i) at (-\c*85+\c*0pt, \i*25 pt) {};
					\foreach \i in {3}
						\foreach \j in {0,1,2}
							\node[vtx] (B\c\i\j) at (-\c*85-\c*40pt, \i*25+\j*5-5 pt) {};
					\foreach \i in {1,4}
						\foreach \j in {1}
							\node[vtx] (B\c\i\j) at (-\c*85-\c*40pt, \i*25+\j*5-5 pt) {};
					\foreach \i in {0,2}
						\foreach \j in {0,2}
							\node[vtx] (B\c\i\j) at (-\c*85-\c*40pt, \i*25+\j*5-5 pt) {};
					\foreach \i in {0,1}
						\foreach \j in {0,1}
							\node[vtx] (C\c\i\j) at (-\c*85+\c*40pt, \i*60+\j*30+5 pt) {};
					\draw [rounded corners=5pt] (-\c*85-\c*8pt, -10pt) rectangle (-\c*85+\c*8pt, 110pt);
					\draw [rounded corners=3pt] (-\c*85-\c*45pt, 39pt) rectangle (-\c*85-\c*35pt, 61pt);
					\draw [rounded corners=4pt] (-\c*85+\c*33pt, -5pt) rectangle (-\c*85+\c*47pt, 45pt);
					\draw [rounded corners=4pt] (-\c*85+\c*33pt, 55pt) rectangle (-\c*85+\c*47pt, 105pt);

					\node at (-\c*85-\c*40pt, -20pt) {\scriptsize \ifthenelse{\c=1}{${}^+V^l$}{${}^-V^l$}};
					\node at (-\c*85+\c*0 pt, -20pt) {\scriptsize \ifthenelse{\c=1}{${}^+V^m$}{${}^-V^m$}};
					\node at (-\c*85+\c*40pt, -20pt) {\scriptsize \ifthenelse{\c=1}{${}^+V^r$}{${}^-V^r$}};
					\foreach \i in {1,2,3,4,5}
						\node at (-\c*85-\c*55pt, \i*25-25 pt) {\scriptsize \ifthenelse{\c=1}{${}^+V^l_\i$}{${}^-V^l_\i$}};
					\node at (-\c*85+\c*57pt, 0pt) {\scriptsize \ifthenelse{\c=1}{${}^+V^r_1$}{${}^-V^r_1$}};
					\node at (-\c*85+\c*57pt, 100pt) {\scriptsize \ifthenelse{\c=1}{${}^+V^r_2$}{${}^-V^r_2$}};

					\draw (B\c41) -- (A\c4) -- (C\c10);
					\draw (A\c4) -- (C\c11);

					\draw (B\c22) -- (A\c2) -- (C\c00);
					\draw (B\c20) -- (A\c2);
				}
				\draw [dashed, rounded corners=3pt] (-154pt, -30pt) rectangle (-15pt, 120pt);
				\draw [dashed, rounded corners=3pt] (154pt, -30pt) rectangle (15pt, 120pt);
				\node at (-85pt, -40pt) {\scriptsize copy 1};
				\node at (85pt, -40pt) {\scriptsize copy 2};

				\node [vtx, label={\scriptsize $t_2$}] (t2) at (0pt, 70pt) {};
				\node [vtx, label={\scriptsize $t_1$}] (t1) at (0pt, 30pt) {};
				\draw [line width=1.5pt] (C-110) -- (t1);
				\draw [line width=1.5pt] (C-111) -- (t1);
				\draw [line width=1.5pt] (C100) -- (t1);
				\draw [line width=1.5pt] (C101) -- (t1);
				\draw [line width=1.5pt] (C-100) -- (t2);
				\draw [line width=1.5pt] (C-101) -- (t2);
				\draw [line width=1.5pt] (C110) -- (t2);
				\draw [line width=1.5pt] (C111) -- (t2);

				\node [vtx, label={\scriptsize $s_2$}] (s2) at (180pt, 40pt) {};
				\node [vtx, label={\scriptsize $s_1$}] (s1) at (-180pt, 40pt) {};
				\draw (s2) edge[line width=1.5pt, out=180, in=310] (A-12);
				\draw (s1) edge[line width=1.5pt, out=0, in=230] (A12);
			\end{tikzpicture}
			\caption{}
		\end{subfigure}
		\begin{subfigure}{0.5\textwidth}
			\centering
			\begin{tikzpicture}[vtx/.style={draw, circle, inner sep=0, minimum size=3pt}]
				\node [vtx, label={\scriptsize $s_1$}] (s1) at (0,20pt) {};
				\node [vtx, label={\scriptsize $s_2$}] (s2) at (0,0) {};
				\node [vtx, label={\scriptsize $t_1$}] (t1) at (70pt,20pt) {};
				\node [vtx, label={\scriptsize $t_2$}] (t2) at (70pt,0) {};
				\draw (s1) edge[decorate, decoration=snake, thick] (t1);
				\draw (s2) edge[decorate, decoration=snake, thick] (t2);
				\node [vtx, label={\scriptsize $s_1$}] (s1a) at (150pt,20pt) {};
				\node [vtx, label={\scriptsize $s_2$}] (s2a) at (150pt,0) {};
				\node [vtx, label={\scriptsize $t_1$}] (t1a) at (220pt,20pt) {};
				\node [vtx, label={\scriptsize $t_2$}] (t2a) at (220pt,0) {};
				\draw (s1a) edge[decorate, decoration=snake, thick] (t2a);
				\draw (s2a) edge[decorate, decoration=snake, thick] (t1a);

				\node at (-15pt, 12pt) {\scriptsize (i)};
				\node at (135pt, 12pt) {\scriptsize (ii)};
			\end{tikzpicture}
			\caption{}
		\end{subfigure}
	\end{center}
	\caption{Subfigure (a) demonstrates one block. Subfigure (b) shows the two possible cases of the connectivity of $s_1,s_2,t_1,t_2$ within the block.}\label{fig_hard_2}
\end{figure}
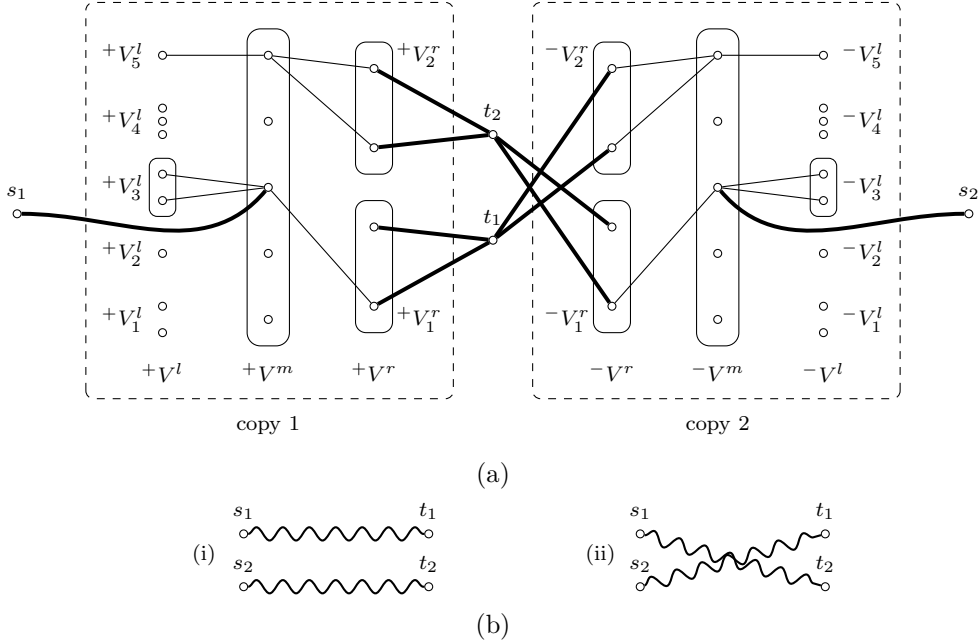

\subsection{Universal relation lower bound}
The previous lower bound for $\ur$~\cite{KNPWWY17} uses an information theoretic argument.\footnote{\cite{KNPWWY17} provided two proofs, we only discuss their first proof here.}
Roughly speaking, the goal is to show that many elements of $S$ can be reconstructed from Alice's message $\pi$ (possibly given some other information about $S$).
Then it would imply that $\pi$ contains lots of information, and thus, it has to be long.
As a demonstration of the argument, let us assume for now, that the protocol always succeeds (i.e., error probability $\delta=0$).
Given Alice's message $\pi$, the reconstruction algorithm can set $T=\emptyset$ and simulate Bob.
Bob returns an element $x_1$ in $S\setminus \emptyset =S$, recovering one element.
Next, it sets $T=\{x_1\}$, and simulates Bob again, which returns $x_2\in S\setminus\{x_1\}$.
Then, it sets $T=\{x_1,x_2\}$, and so on.
This procedure reconstructs the whole set $S$ from $\pi$.
Therefore, the message length must be at least $\Omega(U)$.

The above exact argument breaks when $\delta>0$. 
In the first round, we set $T=\emptyset$ and Bob returns an element $x_1\in S$ with probability $1-\delta$.
However, in the second round where $T=\{x_1\}$, the protocol no longer succeeds with probability $1-\delta$, since we are using the \emph{same} randomness in both rounds.
In the other words, we have to condition on the randomness leading to a first-round output of $x_1$, which distorts its distribution.
Since such an event may have probability as low as $1/|S|$, conditioning on it could significantly affect the error probability.
To resolve this issue, \cite{KNPWWY17} applies the following strategy.
In the second round, instead of setting $T=\{x_1\}$, we also ``mix'' another $\alpha$-fraction of the remaining elements of $S$ into $T$ for some $\alpha\in(0,1)$.
That is, we take a random subset of $S\setminus \{x_1\}$ of size $\alpha\cdot |S|$ and give it to the reconstruction algorithm for free.
The algorithm sets $T$ to be the union of $\{x_1\}$ and this subset, and simulates Bob.
In this way, the condition becomes more mild -- instead of conditioning on the randomness leading to a first-round output of $x_1$, we only condition on the randomness leading to a first-round output \emph{that is in $T$}.
It turns out that by mixing in an $\alpha$-fraction of the remaining elements into $T$ in each round for $\alpha=1/\log (1/\delta)$, one can ensure that the later rounds succeed with high probability.
This argument can therefore be applied for $\Omega(\log (1/\delta)\log U)$ rounds.
Beyond the elements that are given, the algorithm reconstructs $\Omega(\log (1/\delta)\log U)$ extra elements in $S$ in expectation.
It implies that the message length must be at least $\Omega(\log (1/\delta)\log^2 U)$.

\subsection{Lower bound for decision version}
Recall that in the decision version, Bob does not only get $T$, he also gets a bipartition $(P_1,P_2)$ of $[U]\setminus T$ such that $S\setminus T$ is a subset of either $P_1$ or $P_2$.
Therefore, in order to simulate Bob, we must give the reconstruction algorithm a valid bipartition.
This can be deadly -- the bipartition $(P_1,P_2)$ contains at least $O(|S|)$ bits of information about $S$, whereas Bob's output only contains one bit.
Hence, we could at most recover $O(\log U)$ bits from Alice's message $\pi$ before giving away the entire set $S$, which only has $O(|S|\log U)$ bits.

The key component in our lower bound proof is to analyze the information that can be learned from $\pi$, \emph{without being given a valid partition $(P_1, P_2)$}.
For simplicity, let us assume $\delta=0$ and $T=\emptyset$ for now, i.e., let us focus on the first round in the previous argument for zero-error protocols.
Given Alice's message, we can enumerate \emph{all} possible partitions $(P_1, P_2)$ of $[U]$, and simulate Bob on them.
Suppose for a partition $(P_1, P_2)$, Bob returns that $S\setminus T(=S)$ is a subset of $P_1$.
Although we have no way to verify whether it is even a valid input, Bob's output at least tells us that $S$ \emph{cannot} be a subset of $P_2$.
Since if $S\subseteq P_2$, $(P_1, P_2)$ would be a valid partition, in which case, Bob has to output $P_2$.
Thus, for \emph{every} partition $(P_1, P_2)$, we rule out some possibilities for set $S$ by simulating Bob.
The key question here is how much information we can learn by simulating Bob on all partitions and $T=\emptyset$.

Suppose we could show that if all remaining possibilities for $S$ contain some particular element $x_1$, i.e., we have learned that $x_1$ must be in $S$, then we could proceed as in the previous argument.
However, this is not always the case.
An easy counterexample is that for some integer $k>1$, Alice picks $k$ elements from $S$ and another $k-1$ elements from $[U]$, and sends the set $W$ of these $2k-1$ elements to Bob (without annotating which ones are from $S$).
Then for $T=\emptyset$ and any $(P_1,P_2)$, Bob can just output the part that contains at least $k$ elements from $W$.
This part must contain at least one element from $S$, and by the assumption that either $S\subseteq P_1$ or $S\subseteq P_2$, it must contain $S$.
However, if we apply the above strategy enumerating all possible $(P_1,P_2)$ and simulating Bob, the remaining possibilities for $S$ will simply be all $S$ that contain at least $k$ elements from $W$.
There is not an element $x_1$ that is contained in all remaining possibilities for $S$.

However, this counterexample is not a bad case for the whole argument, because by telling the reconstruction algorithm which of the $k$ elements in $W$ belong to $S$ using $O(k)$ extra bits, the algorithm can recover $k$ elements in $S$, which is worth $k\log U$ bits of information.
Then the previous argument still works.
The main technical lemma in this paper is a structural result that asserts this is essentially the only possible type of counterexamples (see Lemma~\ref{lem_intersect}).
\begin{lemma}[main technical lemma, informal]\label{lem_tech_informal}
	Fix any deterministic protocol.
	Suppose for some collection $\cS$ of Alice's inputs, Alice sends the same message $\pi$ on all $S\in\cS$, and Bob is able to compute $\urdec$ for a random $S\in\cS$ and $T=\emptyset$ with probability at least $3/4$, then there must exist one set $S^*\in\cS$ and some integer $k\geq 1$ such that $S^*$ has intersection size $k$ with at least $\exp(-\Theta(k))$-fraction of the sets $S\in\cS$.
\end{lemma}

In the other words, maybe not all sets $S$ consistent with $\pi$ contain the same element $x_1$, but the lemma implies that there must exist a not-so-small fraction of the sets that contain \emph{many} same elements.
This is because by averaging, at least $|S^*|^{-k}\cdot \exp(-\Theta(k))$-fraction of the sets has the \emph{same} intersection of size $k$ with $S^*$.
When $|S^*|\ll U$, this imposes a structure on the collection of sets consistent with $\pi$.
That is, a $(\gg U^{-k})$-fraction of sets contain the same $k$ elements.
We can approximately view it as ``describing'' $k$ elements using $\ll k\log U$ bits.
This lemma also extends to $T\neq \emptyset$.
This allows us to mimic the previous argument.

We first apply Yao's minimax principle to fix the randomness of the protocol.
Our argument starts with the collection $\cS$ of all $S$ on which Alice sends message $\pi$, and $T=\emptyset$.
Then we
\begin{compactitem}
	\item[(a)] apply Lemma~\ref{lem_tech_informal} and find $k$ elements such that $(\gg U^{-k})$-fraction of $\cS$ contain all of them, add those $k$ elements to $T$, and remove all sets in $\cS$ that \emph{do not} contain $T$ (corresponding to recovering elements from $S$ in the previous argument);
	\item[(b)] next pick $\alpha (|S|-|T|)$ (we only consider $S$ of the same size) \emph{random} elements from $[U]\setminus T$, add those elements to $T$, and remove all sets in $\cS$ that {do not} contain $T$ (corresponding to ``mixing'' in $\alpha$-fraction of random remaining elements in $S$).
\end{compactitem} 
We repeatedly apply these two steps, and eventually we have restricted all sets in $\cS$ to contain $|S|$ specific elements.
That is, the final size of $\cS$ can be at most $1$.
Similar to the previous argument, we can show that by mixing in a random $\alpha$-fraction each time, the average success probability of $\cS$ in the later rounds will be at least $3/4$, allowing us to apply Lemma~\ref{lem_tech_informal}.

To see why this argument implies a lower bound on $|\pi|$, observe that in step (b), the size of $\cS$ drops as expected -- by a factor of $\binom{|S|-|T|}{\alpha (|S|-|T|)}/\binom{U-|T|}{\alpha (|S|-|T|)}$, the probability that a set $S$ contains $\alpha (|S|-|T|)$ random elements outside $T$. 
If the size of $\cS$ also dropped as expected in step (a), then in the whole process, the size of $\cS$ would have dropped by the expected factor of $\binom{U}{|S|}^{-1}$, the probability that a set $S$ contains $|S|$ random elements.
Combining it with the final size of $\cS$ being at most $1$, we would only have obtained a trivial upper bound of $\binom{U}{|S|}$ on the initial size of $\cS$.
But in step (a), the actual drop of the size of $\cS$ is much slower.
Thus, in total, the size of $\cS$ dropped by a factor of $\gg \binom{U}{|S|}^{-1}$, implying a non-trivial upper bound of $\ll\binom{U}{|S|}$ on the initial size of $\cS$.
Recall that the initial $\cS$ was the collection of $S$ on which Alice sends $\pi$.
That means there must be $\ll\binom{U}{|S|}$ different inputs that can have Alice send the same message $\pi$.
However, this argument applies to \emph{all} messages $\pi$.
Thus, we must have many different messages in order to cover all $\binom{U}{|S|}$ inputs $S$, implying a lower bound on $|\pi|$.

The actual proof is slightly different due to technical reasons, see the next section.


\section{Lower Bound for $\urdec$}\label{sec_urdec_lb}

	Recall that in the $\urdec$ problem, Alice gets a set $S\subseteq [U]$, Bob gets a proper subset $T\subsetneq S$, as well as a partition $(P_1, P_2)$ of $[U]\setminus T$.
	It is guaranteed that either $S\setminus T\subseteq P_1$, or $S\setminus T\subseteq P_2$.
	In the communication game, Alice sends one message $\pi$ to Bob, and Bob must decide whether $P_1$ or $P_2$ contains $S\setminus T$ with probability at least $1-\delta$.
	In this section, we prove the following lower bound for the $\urdec$ problem.
	\newcommand{\lemurlbdecisioncont}{For any $U$ and $\delta$ such that $\exp(-U^{1/4})<\delta<1/\log^4 U$, there is an input distribution $\cD_{\urdec}$ such that any one-way communication protocol that succeeds with probability at least $1-\delta$ on a random instance sampled from $\cD_{\urdec}$ must have expected communication cost at least $\Omega(\log(1/\delta)\log^2 U)$.}
	\begin{lemma}\label{lem_ur_lb_decision}
		\lemurlbdecisioncont
	\end{lemma}

	\paragraph{Hard distribution $\cD_{\urdec}$.} 
	Without loss of generality, assume $U$ is a perfect cube.
	Let $|S|=m=U^{1/3}$, and we view $[U]$ as $m$ disjoint \emph{blocks} of size $B=U^{2/3}$.
	Alice's input set $S$ is a uniformly random set of size $m$, with one element from each block.
	Let $\alpha=\frac{16}{\log 1/\delta}$, and $t_r=\lceil m\cdot (1-(1-\alpha)^r)+2r\rceil$ for $r=0,1,\ldots,R-1$ be all possible sizes of $T$, where $R=\lfloor \frac{1}{16\alpha}\log m\rfloor$.
	Bob's input set $T$ is a uniformly random subset of $S$ of size $t_r$, for a uniformly random $r$ in $\{0,\ldots,R-1\}$.
	Finally, we put the whole set $S\setminus T$ in either $P_1$ or $P_2$ randomly, and then put each element in $[U]\setminus S$ randomly and independently in $P_1$ or $P_2$, i.e., $(P_1,P_2)$ is a uniformly random partition of $[U]\setminus T$ conditioned on $S\setminus T\subseteq P_1$ or $S\setminus T\subseteq P_2$.

	Note that this distribution is valid when $\delta>\exp(-U^{1/4})$ and $U$ sufficiently large.
	In this case, we have $$t_r<m(1-(1-\alpha)^R)+2R<m-m^{7/8}+O(m^{3/4}\log m)<m,$$ and $t_r\geq 0$.
	Therefore, $T$ is always a proper subset of $S$.
	Also observe that
	\[
		t_{r+1}-t_r\geq m((1-\alpha)^r-(1-\alpha)^{r+1})+1\geq 1.
	\]

	\bigskip

	Suppose there is a randomized communication protocol with error probability at most $\delta$ and expected cost at most $C$.
	By Markov's inequality and union bound, we can fix the randomness of the protocol such that under $\cD_{\urdec}$, the error probability is at most $2\delta$ and the expected communication cost is at most $2C$.
	Thus, we may assume the protocol is deterministic.

	Fix such a deterministic protocol.
	By Markov's inequality and union bound again, for at least $1/2$ of Alice's set $S$, the error probability conditioned on $S$ is at most $8\delta$ and Alice sends a message of length at most $8C$ on $S$.
	Denote this collection of Alice's set by $\cS_{\good}$.
	Thus, $|\cS_{\good}|\geq \frac{1}{2}B^m$.
	Note that this collection could depend on the protocol.

	\bigskip

	To prove the lemma, we fix a message $\pi$ and consider the collection $\cS_0$ of Alice's sets $S\in\cS_{\good}$ on which Alice sends $\pi$.
	We pick $\pi$ that maximizes $|\cS_0|$, hence, $|\cS_0|\geq \frac{1}{2}B^m\cdot 2^{-8C}$.
	We will show that $\cS_0$ has to be small, which will imply a lower bound on $C$.

	To this end, we will describe a random process that generates a sequence of nested collections $\cS_0\supset \cS_1\supset \cS_2\supset\cdots\supset \cS_I$ and a sequence of sets $T_0\subset T_1\subset T_2\subset\cdots\subset T_I$ such that each $T_i$ is one possible input set for Bob, and $T_0=\emptyset$.
	For each $i$, all $S\in\cS_i$ will contain $T_i$.
	Clearly, $|\cS_i|$ is at most $B^{m-|T_i|}$ for every $i$.
	We will then show that $|\cS_i|/B^{m-|T_i|}$ increases rapidly as $i$ increases.
	Combining it with the fact that $|\cS_I|/B^{m-|T_I|}\leq 1$, we obtain that $|\cS_0|/B^{m-|T_0|}=|\cS_0|/B^m$ must be very small.

	\paragraph{Random process $\cA$.} Now let us describe the random process $\cA$ (see Figure~\ref{fig_ca}).
	To initialize, we fix a message $\pi$ of length at most $8C$, which is sent by Alice on the most number of sets $S\in\cS_{\good}$.
	Then let $\cS_0$ be all such sets, and let $T_0$ be the empty set, i.e., $|T_0|=t_{r_0}$ for $r_0=0$.
	Next, we iteratively generate collections $\cS_i$ and sets $T_i$ of size $t_{r_i}$ for some $r_i<R$.

	In round $i$, we construct $\cS_{i+1}$ and $T_{i+1}$ from $\cS_i$ and $T_i$.
	We first check if there are many pairs of sets in $\cS_i$ that intersect \emph{outside} $T_i$ (recall that all sets in $\cS_i$ contain $T_i$).
	More specifically, we try to find a $k_i\geq 1$ such that there are at least $\frac{|\cS_i|^2}{4\cdot 2^{2k_i}}$ pairs of sets in $\cS_i$ intersect on $k_i$ elements outside $T_i$.
	If such $k_i$ does not exist, the random process aborts (and it fails).
	Otherwise, we fix one such $k_i$, and by averaging, there must be a set $S^*$ such that it intersects at least $\frac{|\cS_i|}{4\cdot 2^{2k_i}}$ sets in $\cS_i$ on $k_i$ elements outside $T_i$.
	Again by averaging, there must be a subset $\Delta T\subseteq S^*$, of size $k_i$ and disjoint from $T_i$, such that at least $\frac{|\cS_i|}{4\cdot 2^{2k_i}\cdot m^{k_i}}$ sets $S\in\cS_i$ have $(S\cap S^*)\setminus T_i=\Delta T$.
	In particular, they all contain $T_i\cup \Delta T$.
	
	We then fix any such $S^*$ and $\Delta T$.
	$\Delta T$ will be added to $T_{i+1}$.
	The next set $T_{i+1}$ will have size $t_{r_{i+1}}$ for $r_{i+1}=r_i+k_i$.
	Observe that $|T_{i+1}|=t_{r_{i+1}}$ is at least $|T_i\cup\Delta T|=t_{r_i}+k_i$.
	Then we pick $t_{r_{i+1}}-t_{r_i}-k_i$ \textbf{random} blocks $\cB_i$ that are disjoint from $T_i\cup \Delta T$.
	For each block in $\cB_i$, we pick one element to add to $T_{i+1}$, and denote this set of $t_{r_{i+1}}-t_{r_i}-k_i$ elements by $\Delta T'$.
	By averaging, there exists such a set $\Delta T'$ such that at least $\frac{|\cS_i|}{4\cdot 2^{2k_i}\cdot m^{k_i}\cdot B^{t_{r_{i+1}}-t_{r_i}-k_i}}$ sets in $S\in\cS_i$ have $(S\cap S^*)\setminus T_i=\Delta T$ and $S\supseteq \Delta T'$.
	In particular, they all contain $T_i\cup\Delta T\cup\Delta T'$.
	We fix any such $\Delta T'$.
	Finally, let $T_{i+1}=T_i\cup \Delta T\cup \Delta T'$, and let $\cS_{i+1}$ be a collection of any $\frac{|\cS_i|}{4\cdot 2^{2k_i}\cdot m^{k_i}\cdot B^{t_{r_{i+1}}-t_{r_i}-k_i}}$ sets in $\cS_i$ that contain $T_{i+1}$.

	We repeat this process until $r_{i+1}\geq R$, in which case, $t_{r_{i+1}}$ becomes undefined.
	Then we do not sample random blocks, and simply let $T_{i+1}=T_i\cup \Delta T$ and $\cS_{i+1}$ be the collection of any $\frac{|\cS_i|}{4\cdot 2^{2k_i}\cdot m^{k_i}}$ sets in $\cS_i$ that contain $T_{i+1}$, and end the process.
	To ensure the random process is well-defined, for any step that ``fixes any such $X$'', we mean fixing $X$ to the lexicographically smallest.

	\bigskip
	Note that the only random part in the whole process is in sampling the $t_{r_{i+1}}-t_{r_i}-k_i$ random blocks $\cB_i$ in each round.
	The selection of $\pi$, $k_i,S^*,\Delta T$ and $\Delta T'$ after $\cB_i$ is sampled is deterministic.
	The ``saving'' of each round comes from $\Delta T$: it increases the size of $T_i$ by $k_i$ while $|\cS_i|$ is only reduced by a factor of $\frac{1}{4\cdot (4m)^{k_i}}$ (rather than $B^{k_i}$).
	As we will see later, the elements from random blocks $\cB_i$ ensure the existence of such $k_i$ in later rounds with high probability.

	To avoid ambiguity in the terminology, \emph{error} is only used when referring to the protocol outputting a wrong answer, and \emph{failure} is only used when referring to the random process aborting before reaching $r_i\geq R$.

	\begin{figure}
	\begin{code}{rp}[Random process $\cA$\!\!][][]
		\item find $\pi$ that maximizes $|\{S\in\cS_{\good}: \textrm{Alice sends $\pi$ on input $S$}\}|$
		\item let $\cS_0:=\{S\in\cS_{\good}: \textrm{Alice sends $\pi$ on input $S$}\}$
		\item let $r_0:=0, T_0:=\emptyset$
		\item let $i:=0$
		\item repeat
		\item \quad if there is no $k_i\geq 1$ such that $\left|\{S_1,S_2\in \cS_i: |(S_1\cap S_2)\setminus T_i|=k_i\}\right|\geq \frac{|\cS_i|^2}{4\cdot 2^{2k_i}}$
		\item \quad\quad the process fails, abort
		\item \quad find any $k_i$, $S^*$, and $\Delta T$ of size $k_i$ such that $\left|\{S_2\in \cS_i: (S^*\cap S_2)\setminus T_i=\Delta T\}\right|\geq \frac{|\cS_i|}{4\cdot 2^{2k_i}\cdot m^{k_i}}$
		\item \quad let $r_{i+1}:=r_i+k_i$
		\item \quad if $r_{i+1}<R$
		\item\label{rp_add_rand_beg} \quad\quad pick $t_{r_{i+1}}-t_{r_i}-k_i$ random blocks $\cB_i$ that are disjoint from $T_i\cup \Delta T$
		\item \quad\quad find any $\Delta T'$ consisting of exactly one element from each block in $\cB_i$, such that \\ 
		\phantom{}\quad\quad\quad$\left|\{S_2\in \cS_i: (S^*\cap S_2)\setminus T_i=\Delta T, S_2\supseteq \Delta T'\}\right|\geq \frac{|\cS_i|}{4\cdot 2^{2k_i}\cdot m^{k_i}\cdot B^{t_{r_{i+1}}-t_{r_i}-k_i}}$
		\item \quad\quad let $T_{i+1}:=T_i\cup \Delta T\cup \Delta T'$
		\item\label{rp_add_rand_end} \quad\quad let $\cS_{i+1}$ be the collection of any $\frac{|\cS_i|}{4\cdot 2^{2k_i}\cdot m^{k_i}\cdot B^{t_{r_{i+1}}-t_{r_i}-k_i}}$ sets $S\in\cS_i$ that contain $T_{i+1}$
		\item \quad else
		\item \quad\quad let $T_{i+1}:=T_i\cup \Delta T$
		\item \quad\quad let $\cS_{i+1}$ be the collection of any $\frac{|\cS_i|}{4\cdot 2^{2k_i}\cdot m^{k_i}}$ sets $S\in\cS_i$ that contain $T_{i+1}$
		\item \quad $i:=i+1$
		\item until $r_i\geq R$
		\item denote the final $i$ by $I$
	\end{code}
	\caption{Random process $\cA$.}\label{fig_ca}
	\end{figure}

	The key property of $\cA$ is that it does not always fail.
	\begin{lemma}\label{lem_failure_prob}
		The probability that $\cA$ fails is at most $1/2$ as long as $\exp(-U^{1/4})<\delta<1/\log^4 U$.
	\end{lemma}
	We will prove the lemma in the next subsection.
	Let us first show that it implies the claimed $\urdec$ lower bound.
	\begin{proof}[Proof of Lemma~\ref{lem_ur_lb_decision}]
		By Lemma~\ref{lem_failure_prob}, $\cA$ does not always fail.
		We draw a sample from $\cA$ conditioned on succeeding, and obtain collections $\cS_0,\ldots,\cS_I$.
		By construction, we have
		\begin{align*}
			|\cS_I|&=|\cS_0|\cdot\left(\prod_{i=0}^{I-2}\frac{1}{4\cdot 2^{2k_i}\cdot m^{k_i}\cdot B^{t_{r_{i+1}}-t_{r_i}-k_i}}\right)\cdot\frac{1}{4\cdot 2^{2k_{I-1}}\cdot m^{k_{I-1}}} \\
			&= |\cS_0|\cdot4^{-I}\left(\frac{B}{4m}\right)^{\sum_{i=0}^{I-1}k_i}\cdot B^{-k_{I-1}}\prod_{i=0}^{I-2}\frac{1}{B^{t_{r_{i+1}}-t_{r_i}}} \\
			&\geq |\cS_0|\cdot4^{-I}\left(\frac{B}{4m}\right)^R\cdot B^{-(t_{r_{I-1}}+k_{I-1})},
		\end{align*}
		where the last inequality uses the fact that $\sum_{i=0}^{I-1}k_i=r_I\geq R$ and $B>4m$.
		Then by the fact that $I\leq R$ and $|T_I|=t_{r_{I-1}}+k_{I-1}$, we have
		\[
			|\cS_I|\geq |\cS_0|\cdot\left(\frac{B}{16m}\right)^R\cdot B^{-|T_I|}.
		\]
		On the other hand, $|\cS_I|\leq B^{m-|T_I|}$.
		Thus,
		\[
			|\cS_0|\leq |\cS_I|\cdot \left(\frac{16m}{B}\right)^R\cdot B^{|T_I|}\leq \left(\frac{16m}{B}\right)^R\cdot B^{m}.
		\]
		However, by averaging, $|\cS_0|\geq |\cS_{\good}|\cdot 2^{-8C}\geq B^m\cdot 2^{-8C-1}$.
		Therefore, we have $2^{-8C-1}\leq \left(\frac{16m}{B}\right)^R$, which simplifies to
		 $$C\geq \Omega(R\log (B/16m))=\Omega(\log (1/\delta)\log^2 U),$$ proving the lemma.
	\end{proof}

\subsection{Failure probability of $\cA$}\label{sec_fail_prob}
	Now let us bound the failure probability of $\cA$, proving Lemma~\ref{lem_failure_prob}.
	To this end, we will first show that for any round $i$ and any $\cS_i,T_i$, if the conditional error probability of the protocol, under input distribution $\cD_{\urdec}$ conditioned on $S\in\cS_i$ and $T=T_i$, is at most $1/4$, then $\cA$ does not fail in this round.
	\begin{lemma}\label{lem_intersect}
		If conditioned on $S\in\cS_i$ and $T=T_i$, the error probability is at most $1/4$, then we must have
		\[
			\sum_{S_1,S_2\in \cS_i} \left(2^{|(S_1\cap S_2)\setminus T_i|}-1\right)\geq \frac{|\cS_i|^2}{4},
		\]
		and consequently, there exists some $k_i\geq 1$ such that
		\[
			|\{S_1,S_2\in\cS_i:|(S_1\cap S_2)\setminus T_i|= k_i\}|\geq \frac{|\cS_i|^2}{4\cdot 2^{2k_i}}.
		\]
	\end{lemma}

	Then we will upper bound the probability that $\cA$ generates $\cS_i, T_i$ whose conditional error probability is more than $1/4$, by applying the following lemma.
	Fix $k_0,\ldots,k_{i-1}$, which determines $r_0,\ldots,r_i$, and consider the distribution of $\cS_i$ and $T_i$ (which has size $t_{r_i}$) induced by $\cA$ conditioned on $k_0,\ldots,k_{i-1}$.
	The lemma states that for \emph{any} $S\in\cS_{\good}$, and \emph{any} $T\subset S$ of size $t_{r_i}$, the probability that $T_i=T$ conditioned on $\cS_i\ni S$ and $k_0,\ldots,k_{i-1}$ is at most than $2^{6/\alpha}\cdot\binom{m}{t_{r_i}}^{-1}$, i.e., the probability of any set $T$ conditioned on $S\in\cS_i$ can increase by at most a factor of $2^{6/\alpha}$ (compared to the uniform distribution over subsets of $S$ of size $t_{r_i}$).
	\begin{lemma}\label{lem_singleton_prob}
		Fix any $k_0,\ldots,k_{i-1}$, which determines $r_0,\ldots,r_i$, such that $r_i<R$.
		For any $S\in\cS_{\good}$ such that $\Pr[S\in\cS_i\mid k_0,\ldots,k_{i-1}]>0$ and any $T\subset S$ of size $t_{r_i}$, we must have
		\[
			\Pr[T_i=T\mid S\in\cS_i,k_0,\ldots,k_{i-1}]\leq \frac{2^{6/\alpha}}{\binom{m}{t_{r_i}}},
		\]
		over the randomness of $\cA$.
	\end{lemma}

	The above two lemmas together imply the claimed upper bound on the failure probability of $\cA$.
	\begin{proof}[Proof of Lemma~\ref{lem_failure_prob}]
		By the definition of $\cS_{\good}$, for any $S\in\cS_{\good}$, the error probability conditioned on $S$ is at most $8\delta$. 
		Recall that in $\cD_{\urdec}$, the size of $T$ is $t_r$ for a uniformly random $r=0,\ldots,R-1$.
		It implies that for any fixed $S\in\cS_{\good}$ and fixed $r$, the error probability conditioned on $S$ and $|T|=t_r$ is at most $8R\delta$.
		Now instead of sampling a random subset $T$, suppose we replace the conditional distribution of $T$ conditioned on $S$, by the distribution of $T_i$ generated by the random process \emph{conditioned on} $S\in\cS_i$ and $k_0,\ldots,k_{i-1}$.
		Then by Lemma~\ref{lem_singleton_prob}, for any $k_0,\ldots,k_{i-1}$ such that $r_i<R$, and any $S\in\cS_{\good}$ such that $\Pr[S\in \cS_i\mid k_0,\ldots,k_{i-1}]>0$, the expected error probability of the protocol conditioned on $S$ and $T$ is at most $2^{6/\alpha+3}R\delta$.
		Note that $|\cS_i|$ is \emph{fixed} given $k_0,\ldots,k_{i-1}$, hence, the expected error probability conditioned on $S\in\cS_i$ and $T=T_i$ is also at most $2^{6/\alpha+3}R\delta$:
		\begin{align*}
			&\phantom{=}\,\,\E_{\cS_i,T_i\mid k_0,\ldots,k_{i-1}}[\Pr[\textrm{err}\mid S\in\cS_i,T=T_i]] \\
			&=\frac{1}{|\cS_i|}\cdot\E_{\cS_i,T_i\mid k_0,\ldots,k_{i-1}}\left[\sum_{S\in\cS_i}\Pr[\textrm{err}\mid S,T=T_i]\right] \\
			&=\frac{1}{|\cS_i|}\cdot\E_{\cS_i,T_i\mid k_0,\ldots,k_{i-1}}\left[\sum_{S\in\cS_{\good}}\mathbf{1}_{S\in\cS_i}\cdot\Pr[\textrm{err}\mid S,T=T_i]\right] \\
			&=\sum_{S\in\cS_{\good}}\frac{1}{|\cS_i|}\cdot\E_{\cS_i,T_i\mid k_0,\ldots,k_{i-1}}\left[\mathbf{1}_{S\in\cS_i}\cdot\Pr[\textrm{err}\mid S,T=T_i]\right] \\
			&=\sum_{S\in\cS_{\good}}\frac{1}{|\cS_i|}\cdot \Pr[S\in\cS_i\mid k_0,\ldots,k_{i-1}]\cdot \E_{\cS_i,T_i\mid \cS_i\ni S,k_0,\ldots,k_{i-1}}\left[\Pr[\textrm{err}\mid S,T=T_i]\right] \\
			&\leq\sum_{S\in\cS_{\good}}\frac{1}{|\cS_i|}\cdot \Pr[S\in\cS_i\mid k_0,\ldots,k_{i-1}]\cdot2^{6/\alpha+3}R\delta \\
			&=2^{6/\alpha+3}R\delta.
		\end{align*}
		By Markov's inequality, the probability conditioned on $k_0,\ldots,k_{i-1}$ that the random process generates $\cS_i,T_i$ such that 
		\[
			\Pr[\textrm{err}\mid S\in\cS_i,T=T_i]>1/4
		\]
		is at most $2^{6/\alpha+5}R\delta$.
		Thus, by Lemma~\ref{lem_intersect}, for any $k_0,\ldots,k_{i-1}$ such that $r_i<R$, the probability that $\cA$ fails in round $i$ is at most $2^{6/\alpha+5}R\delta$.
		Averaging over $k_0,\ldots,k_{i-1}$, it implies that the probability that $\cA$ does not fail in first $i-1$ rounds but fails in round $i$ is at most $2^{6/\alpha+5}R\delta$.
		Summing over $i=0,\ldots,R-1$ implies the overall failure probability is at most 
		\[
			2^{6/\alpha+5}R^2\delta.
		\]
		Since $\alpha=\frac{16}{\log 1/\delta}$, and $R\leq \frac{1}{768}\cdot\log (1/\delta)\log U$ and $\delta<1/\log^4 U$, the probability that $\cA$ fails is at most $1/2$.
		This proves the lemma.
	\end{proof}

	In the following, we prove the two remaining lemmas.
	\begin{proof}[Proof of Lemma~\ref{lem_intersect}]
		For each $S\in\cS_i$, conditioned on $S$ and $T=T_i$, by the construction of the hard distribution, $[U]\setminus T$ is randomly partitioned into $(P_1,P_2)$ conditioned on $S\setminus T\subseteq P_1$ or $S\setminus T\subseteq P_2$.
		We first observe that conditioned on $S$ and $T$, the partition restricted to each block is uniform, i.e., the elements in the same block belong to $P_1$ or $P_2$ uniformly and independently.
		This is because each block may have at most one element in $S\setminus T$.
		Moreover, $(P_1,P_2)$ restricted to different blocks is independent of each other, up to switching the order of two parts.
		That is, conditioned on $P_1$ and $P_2$ restricted to first $j$ blocks, the (unordered) set $\{P_1\cap\textrm{block $j+1$}, P_2\cap\textrm{block $j+1$}\}$ is still a uniformly random partition of block $j+1$.

		Therefore, to sample a random input conditioned on $S\in\cS_i$ and $T=T_i$, it is equivalent to do the following:
		\begin{compactenum}
			\item for each block $j$, randomly partition the elements that are not in $T_i$ into $(B_{j,1},B_{j,2})$;
			\item sample a uniformly random $S\in\cS_i$;
			\item pick a random $b\in\{1,2\}$, let $P_b$ be the union over $j$, the part in $\{B_{j,1},B_{j,2}\}$ that contain an element in $S\setminus T_i$ (if no such element in the block, then a random part), let $P_{3-b}$ be the union of the other parts.
		\end{compactenum}

		Thus, conditioned on $\{(B_{j,1},B_{j,2})\}_{j\in [m]}$ in step 1, a partition $(P_1,P_2)$ can be generated only if there exists $a_1,\ldots,a_m\in\{1,2\}$ such that $P_1=\bigcup_{j=1}^m B_{j,a_j}$ (and thus, $P_2=\bigcup_{j=1}^m B_{j,3-a_j}$).
		Moreover, the probability that such a partition is generated (conditioned on step 1) is
		\begin{equation}\label{eqn_prob_p1_p2}
			\frac{A_1+A_2}{|\cS_i|}\cdot 2^{-|T_i|-1},
		\end{equation}
		where $A_1:=|\{S\in\cS_i: S\setminus T_i\subset P_1\}|$ and $A_2:=|\{S\in\cS_i: S\setminus T_i\subset P_2\}|$.
		To see this, with probability $(A_1+A_2)/|\cS_i|$, we pick a set $S$ such that $S\setminus T_i\subset P_1$ or $S\setminus T_i\subset P_2$ in step 2.
		Then with probability $1/2$, we pick the right $b$, and finally, for each block that does not contain an element in $S\setminus T_i$ (i.e., that contains an element in $T_i$), with probability $1/2$, we pick the right part to join $P_b$.

		On the other hand, conditioned on such a partition $(P_1, P_2)$ (and $S\in\cS_i, T=T_i$), the error probability is at least 
		\[
			\frac{\min\{A_1,A_2\}}{A_1+A_2}=\frac{1}{2}\cdot \left(1-\frac{|A_1-A_2|}{A_1+A_2}\right),
		\]
		since Bob outputs an answer based only on $T, P_1, P_2$ and the message, and all sets $S\in\cS_i$ have the same message.
		Hence, no matter which part Bob answers, he makes at least $\min\{A_1,A_2\}$ errors among $A_1+A_2$ possible sets $S$ (and all sets $S$ are chosen with the same probability).

		Combining it with \eqref{eqn_prob_p1_p2}, the error probability conditioned on the partitions $\{(B_{j,1},B_{j,2})\}_{j\in[m]}$ is at least
		\begin{align*}
			&\, \sum_{\stackrel{a_1,\ldots,a_m\in\{1,2\}}{P_1=\bigcup_{j=1}^B B_{j,a_j},P_2=\bigcup_{j=1}^B B_{j,3-a_j}}}\frac{A_1+A_2}{|\cS_i|}\cdot 2^{-|T_i|-1}\cdot \frac{1}{2}\cdot \left(1-\frac{|A_1-A_2|}{A_1+A_2}\right) \\
			=&\, \frac{1}{2}-\sum_{\stackrel{a_1,\ldots,a_m\in\{1,2\}}{P_1=\bigcup_{j=1}^B B_{j,a_j},P_2=\bigcup_{j=1}^B B_{j,3-a_j}}}\frac{A_1+A_2}{|\cS_i|}\cdot 2^{-|T_i|-1}\cdot \frac{1}{2}\cdot \frac{|A_1-A_2|}{A_1+A_2} \\
			=&\, \frac{1}{2}-\sum_{\stackrel{a_1,\ldots,a_m\in\{1,2\}}{P_1=\bigcup_{j=1}^B B_{j,a_j},P_2=\bigcup_{j=1}^B B_{j,3-a_j}}}\frac{|A_1-A_2|}{|\cS_i|}\cdot\frac{1}{2^{|T_i|+2}}.
		\end{align*}
		By taking the expectation over $\{(B_{j,1},B_{j,2})\}_{j\in[m]}$ and switching the order of summation and expectation, we obtain that the error probability conditioned on $S\in\cS_i,T=T_i$ is at least
		\begin{equation}\label{eqn_error_prob}
			\frac{1}{2}-\frac{1}{|\cS_i|\cdot 2^{|T_i|+2}}\cdot \sum_{a_1,\ldots,a_m\in\{1,2\}}\E_{\{(B_{j,1},B_{j,2})\}_{j\in[m]}}\left[|A_1-A_2|\right],
		\end{equation}
		where $A_1=|\{S\in\cS_i:S\setminus T_i\subset P_1\}|, P_1=\bigcup_{j=1}^B B_{j,a_j}$ and $A_2=|\{S\in\cS_i:S\setminus T_i\subset P_2\}|, P_2=\bigcup_{j=1}^B B_{j,3-a_j}$.

		Now, observe that for any sequence $a_1,\ldots,a_m$, the marginal distribution of $P_1$ (or $P_2$) over a random $\{(B_{j,1},B_{j,2})\}_{j\in[m]}$ is simply a uniform subset of $[U]\setminus T_i$.
		By linearity of expectation, the expectation of $A_1$ is equal to
		\[
			\E[A_1]=\frac{|\cS_i|}{2^{|S\setminus T_i|}}=\frac{|\cS_i|}{2^{m-|T_i|}}.
		\]
		Its variance is equal to
		\begin{align*}
			\E[A_1^2]-\E[A_1]^2&=\sum_{S_1,S_2\in \cS_i} \frac{1}{2^{|(S_1\cup S_2)\setminus T_i|}}-\left(\frac{|\cS_i|}{2^{m-|T_i|}}\right)^2\\
			&=\sum_{S_1,S_2\in \cS_i}\left(\frac{1}{2^{2(m-|T_i|)-|(S_1\cap S_2)\setminus T_i|}}-\frac{1}{2^{2(m-|T_i|)}}\right) \\
			&=\frac{1}{2^{2(m-|T_i|)}}\cdot \sum_{S_1,S_2\in \cS_i}\left(2^{|(S_1\cap S_2)\setminus T_i|}-1\right).
		\end{align*}

		\emph{Assuming for contradiction} that the lemma does not hold, i.e., $\sum_{S_1,S_2\in \cS_i} \left(2^{|(S_1\cap S_2)\setminus T_i|}-1\right)<\frac{|\cS_i|^2}{4}$, then the variance is at most
		\[
			\E[(A_1-\E[A_1])^2]=\E[A_1^2]-\E[A_1]^2<\frac{1}{2^{2(m-|T|)}}\cdot \frac{|\cS_i|^2}{4}.
		\]
		Similarly for $A_2$, if the lemma does not hold, then
		\[
			\E[(A_2-\E[A_2])^2]<\frac{1}{2^{2(m-|T|)}}\cdot \frac{|\cS_i|^2}{4}.
		\]
		
		Next, by triangle inequality and the fact that $\E[A_1]=\E[A_2]$,
		\[
			\E[|A_1-A_2|]\leq \E[|A_1-\E[A_1]|]+\E[|A_2-\E[A_2]|].
		\]
		Then by convexity,
		\[
			\E[|A_1-\E[A_1]|]\leq \sqrt{\E[(A_1-\E[A_1])^2]}<\frac{|\cS_i|}{2\cdot 2^{m-|T|}}.
		\]
		and
		\[
			\E[|A_2-\E[A_2]|]\leq \sqrt{\E[(A_2-\E[A_2])^2]}<\frac{|\cS_i|}{2\cdot 2^{m-|T|}}.
		\]
		Hence, $\E[|A_1-A_2|]<\frac{|\cS_i|}{2^{m-|T|}}$.
		Plug it into \eqref{eqn_error_prob}, we obtain that if the lemma does not hold, then the error probability conditioned on $S\in\cS_i,T=T_i$ is \emph{strictly} larger than
		\[
			\frac{1}{2}-\frac{1}{|\cS_i|\cdot 2^{|T|+2}}\cdot 2^m\cdot \frac{|\cS_i|}{2^{m-|T|}}=\frac{1}{2}-\frac{1}{4}=\frac14.
		\]
		It contradicts with the lemma premise that it is at most $1/4$, and hence, we must have
		\[
			\sum_{S_1,S_2\in \cS_i}\left(2^{|(S_1\cap S_2)\setminus T_i|}-1\right)\geq \frac{|\cS_i|^2}{4}.
		\]
		Finally, if for all $k_i\geq 1$, $|\{S_1,S_2\in\cS_i: |(S_1\cap S_2)\setminus T_i|=k_i\}|<\frac{|\cS_i|^2}{4\cdot 2^{2k_i}}$, then the above sum could only be smaller than
		\[
			\sum_{k_i\geq 0}\left(2^{k_i}-1\right)\cdot \frac{|\cS_i|^2}{4\cdot 2^{2k_i}}<\frac{|\cS_i|^2}{4}.
		\]
		This proves the lemma.
	\end{proof}

	It remains to prove Lemma~\ref{lem_singleton_prob}.
	It is similar to Lemma 5 in~\cite{KNPWWY17} and Claim B.3 in~\cite{NY19}.
	\begin{proof}[Proof of Lemma~\ref{lem_singleton_prob}]
		To upper bound the probability that $T_i=T$, first observe that it could only happen if for all $j=0,\ldots,i-1$, all $t_{r_{j+1}}-t_{r_j}-k_j$ randomly chosen blocks in $\cB_j$ contain an element in $T$, because otherwise we would have added some element not in $T$ to set $T_i$.
		By the fact that exactly $t_{r_i}$ blocks contain an element in $T$, this probability is
		\begin{align*}
			\prod_{j=0}^{i-1}\frac{\binom{t_{r_i}-t_{r_j}-k_j}{t_{r_{j+1}}-t_{r_j}-k_j}}{\binom{m-t_{r_j}-k_j}{t_{r_{j+1}}-t_{r_j}-k_j}}&=\prod_{j=0}^{i-1}\frac{(t_{r_i}-t_{r_j}-k_j)!(m-t_{r_{j+1}})!}{(t_{r_i}-t_{r_{j+1}})!(m-t_{r_j}-k_j)!} \\
			&=\frac{t_{r_i}!(m-t_{r_i})!}{m!}\cdot \prod_{j=0}^{i-1}\frac{(t_{r_i}-t_{r_j}-k_j)!(m-t_{r_j})!}{(t_{r_i}-t_{r_j})!(m-t_{r_j}-k_j)!} \\
			&\leq\frac{1}{\binom{m}{t_{r_i}}}\cdot \prod_{j=0}^{i-1}\left(\frac{m-t_{r_j}}{t_{r_i}-t_{r_j}-k_j}\right)^{k_j} \\
			&\leq\frac{1}{\binom{m}{t_{r_i}}}\cdot \prod_{j=0}^{i-1}\left(\frac{m(1-\alpha)^{r_j}-2r_j}{m(1-\alpha)^{r_j}-m(1-\alpha)^{r_i}-2r_j+2r_i-1-k_j}\right)^{k_j} \\
			&\leq\frac{1}{\binom{m}{t_{r_i}}}\cdot \prod_{j=0}^{i-1}\left(\frac{m(1-\alpha)^{r_j}}{m(1-\alpha)^{r_j}-m(1-\alpha)^{r_i}}\right)^{k_j} \\
			&=\frac{1}{\binom{m}{t_{r_i}}}\cdot \prod_{j=0}^{i-1}\left(\frac{1}{1-(1-\alpha)^{r_i-r_j}}\right)^{k_j}.
		\end{align*}
		Since $r_j=k_0+\cdots+k_{j-1}$ for $j=0,\ldots,i$, the last product is
		\begin{align*}
			&\phantom{=}\,\,\prod_{j=0}^{i-1}\left(\frac{1}{1-(1-\alpha)^{r_i-r_j}}\right)^{k_j} \\
			&\leq\prod_{j=0}^{i-1}\prod_{l=0}^{k_j-1}\left(\frac{1}{1-(1-\alpha)^{r_i-(r_j+l)}}\right)\\
			&=\prod_{x=0}^{r_i-1}\frac{1}{1-(1-\alpha)^{r_i-x}}\\
			&\leq\prod_{x=1}^{\infty}\frac{1}{1-(1-\alpha)^{x}} \\
			&=\prod_{x=1}^{\lfloor 1/\alpha\rfloor}\frac{1}{1-(1-\alpha)^{x}}\cdot \prod_{x>\lfloor 1/\alpha\rfloor}\frac{1}{1-(1-\alpha)^{x}}
			\intertext{which, by the fact that $(1-\alpha)^x\leq 1-\frac{1}{2}\alpha x$ when $\alpha x\leq 1$ and the fact that $1/(1-\varepsilon)\leq e^{2\varepsilon}$ when $\varepsilon<1/2$, is}
			&\leq\prod_{x=1}^{\lfloor 1/\alpha\rfloor}\frac{2}{\alpha x}\cdot \prod_{x>\lfloor 1/\alpha\rfloor}e^{2(1-\alpha)^{x}}
			\intertext{which, by the fact that $t!\geq (t/e)^t$, is}
			&\leq\left(\frac{2}{\alpha}\right)^{\lfloor 1/\alpha\rfloor }\left(\frac{e}{\lfloor 1/\alpha\rfloor}\right)^{\lfloor 1/\alpha\rfloor}\cdot\prod_{x\geq 0}e^{2(1-\alpha)^{x}} \\
			&\leq\left(\frac{2e}{1-\alpha}\right)^{1/\alpha}\cdot e^{2/\alpha} \\
			&\leq 2^{6/\alpha}.
		\end{align*}
	\end{proof}
	
\section{Sketch Size Lower Bound}\label{sec_red}

In this section, we prove our main theorem.

\begin{restate}[Theorem~\ref{thm_main}]
	\thmmaincont
\end{restate}

\paragraph{Hard distribution $\cD_{\conn}$.}
We begin by describing the hard instances.
In a hard instance, the graph $G$ consists of $\sqrt{n}$ ``blocks'' of size $\sqrt{n}$, where the $i$-th block consists of vertices labeled from $(i-1)\sqrt{n}+1$ to $i\sqrt{n}$.
To generate $G$, we first independently generate a subgraph $G_i$ for each block.
Each block $i$ has four special vertices $s_1^{(i)}, s_2^{(i)}, t_1^{(i)}, t_2^{(i)}$.
$G_i$ always forms two connected components such that either
\begin{compactenum}[(a)]
	\item $s_1^{(i)}$ and $t_1^{(i)}$ are in one component, $s_2^{(i)}$ and $t_2^{(i)}$ are in the other, or 
	\item $s_1^{(i)}$ and $t_2^{(i)}$ are in one component, $s_2^{(i)}$ and $t_1^{(i)}$ are in the other.
\end{compactenum}
We sample each $G_i$ \emph{independently} from the distribution $\overline{\cD_{\blk}}$, which we will describe in the next subsection.
To complete the construction, we add an edge between $t_1^{(i)}$ and $s_1^{(i+1)}$ and an edge between $t_2^{(i)}$ and $s_2^{(i+1)}$ for $i=1,\ldots,\sqrt{n}$, where block $\sqrt{n}+1$ is block $1$ for simplicity of notations.

\bigskip

To decide if $G$ is connected, let $b_i=0$ if $s_1^{(i)}$ and $t_1^{(i)}$ are in the same component \emph{within $G_i$}, and $b_i=1$ otherwise.
It is easy to verify that the entire graph $G$ is connected \emph{if and only if} $\bigoplus_{i=1}^{\sqrt{n}}b_i=1$.
Intuitively, if the referee can decide the XOR of all $b_i$ with constant probability, then it should be at least able to decide some $b_i$ with probability $1-1/\sqrt{n}$ on average.
In the next subsection, we will show that deciding one $b_i$ with such a small error probability requires sketch size of $\Omega(\log^3 n)$.
\begin{lemma}\label{lem_blk_lb}
	There is a distribution $\overline{\cD_{\blk}}$ such that if a protocol can decide whether $s_1$ connects to $t_1$ or $t_2$ with probability $1-2/\sqrt n$ on a random graph sampled from $\overline{\cD_{\blk}}$, then the average sketch size is at least $\Omega(\log^3 n)$ in expectation.
\end{lemma}
Now, we use an embedding argument to prove Theorem~\ref{thm_main} assuming the lemma.
\begin{proof}[Proof of Theorem~\ref{thm_main}]
Let us first fix a protocol $\cP$ that can decide the connectivity of a random graph $G$ sampled from $\cD_{\conn}$ with error probability at most $1/4$.
Suppose the expected average sketch size is $L$.
By Markov's inequality and union bound, we may fix the random bits of $\cP$ such that the error probability is at most $1/3$ and the expected average sketch size is $4L$.
In the following, we assume that $\cP$ is deterministic.
Observe that no vertex in the graph can simultaneously see edges in more than one block, and thus, every sketch sent to the referee depends only on at most one of the blocks.
Since all $G_i$ are sampled independently, it implies that they must remain independent even \emph{conditioned on} all sketches.

Now, let $\epsilon_i\in[-1/2,1/2]$ be the random variable denoting the \emph{bias} of $b_i$ conditioned on the sketches.
That is, conditioned on all sketches, $s_1^{(i)}$ and $t_1^{(i)}$ are in the same component in $G_i$ with probability $1/2+\epsilon_i$.
In this case, from the view of the referee (i.e., conditioned on all sketches), by the independence of the blocks, the probability that the graph is \emph{not} connected is equal to
\begin{align*}
	\Pr\left[\bigoplus_{i=1}^{\sqrt{n}} b_i=0\right]&=\Pr\left[\bigoplus_{i=2}^{\sqrt{n}} b_i=0\wedge b_1=0\right]+\Pr\left[\bigoplus_{i=2}^{\sqrt{n}} b_i=1\wedge b_1=1\right] \\
	&=\left(\frac12+\epsilon_1\right)\Pr\left[\bigoplus_{i=2}^{\sqrt{n}} b_i=0\right]+\left(\frac12-\epsilon_1\right)\left(1-\Pr\left[\bigoplus_{i=2}^{\sqrt{n}} b_i=0\right]\right) \\
	&=\frac12+2\epsilon_1\cdot\left(\Pr\left[\bigoplus_{i=2}^{\sqrt{n}} b_i=0\right]-\frac12\right) \\
	&=\frac12+(2\epsilon_1)(2\epsilon_2)\cdot\left(\Pr\left[\bigoplus_{i=3}^{\sqrt{n}} b_i=0\right]-\frac12\right) \\
	&=\cdots \\
	&=\frac{1}{2}+\frac{1}{2}\prod_{i=1}^{\sqrt{n}}(2\epsilon_i).
\end{align*}
No matter what the referee outputs, the answer is wrong with probability at least
\[
	\frac{1}{2}-\left|\frac{1}{2}\prod_{i=1}^{\sqrt{n}}(2\epsilon_i)\right|.
\]
Since the overall error probability is at most $1/3$, we have
\[
	\E\left[\left|\prod_{i=1}^{\sqrt{n}}(2\epsilon_i)\right|\right]\geq \frac{1}{3}.
\]
By the fact that all $G_i$ are independent and each sketch depends only on one $G_i$, all $\epsilon_i$ are independent.
Hence, $\prod_{i=1}^{\sqrt{n}}\E[|2\epsilon_i|]\geq \frac13$.
By Markov's inequality and union bound, there exists some $i^*$ such that $\E[|2\epsilon_{i^*}|]\geq 1-\frac{4}{\sqrt{n}}$ and the expected average sketch size of block $i^*$ is at most $8L$.

\bigskip

Next, we embed a random graph $G_{\blk}$ sampled from $\overline{\cD_{\blk}}$ into block $i^*$ and show that by simulating $\cP$, the referee can decide if $s_1$ connects to $t_1$ or $t_2$ with high probability.
We first fix any bijection between the vertex labels of $G_{\blk}$ and the labels of block $i^*$.
Given $G_{\blk}$, each player first maps the labels according to the bijection.
Then for the four special vertices $s_1^{(i^*)},s_2^{(i^*)},t_1^{(i^*)},t_2^{(i^*)}$, they locally add one extra neighbor $t_1^{(i^*-1)}$, $t_2^{(i^*-1)}$, $s_1^{(i^*+1)}$ and $s_2^{(i^*+1)}$ respectively.
Then each vertex computes a sketch of their new neighborhood and sends it to the referee.
The expected average sketch size is at most $8L$ by the definition of $i^*$.
The referee receives sketches from all vertices in block $i^*$, samples the rest of the graph (which is independent of $G_{i^*}$), simulates all other vertices and computes the sketches.
Over the randomness of $G_{\blk}$ as well as the referee's sample of the rest of $G$, the whole graph follows the hard distribution $\cD_{\conn}$.
By the above argument, we have $\E[|2\epsilon_{i^*}|]\geq 1-\frac{4}{\sqrt{n}}$.
Recall that $\epsilon_{i^*}$ is the random variable such that $s_1^{(i^*)}$ and $t_1^{(i^*)}$ are in the same component within $G_{i^*}$ with probability $1/2+\epsilon_{i^*}$ conditioned on the sketches.
Finally, the referee examines the conditional distribution of $G_{\blk}$ conditioned on the sketches, and computes $\epsilon_{i^*}$.
If $\epsilon_{i^*}\geq 0$, the referee outputs ``$s_1$ and $t_1$ are in the same component in $G_{\blk}$'', otherwise it outputs ``$s_1$ and $t_2$ are in the same component''.

The error probability conditioned on the sketches is equal to $\frac{1}{2}-|\epsilon_{i^*}|$, whose expectation is
\[
	\E\left[\frac{1}{2}-|\epsilon_{i^*}|\right]\leq \frac{2}{\sqrt{n}}.
\]
Since this protocol decides if $s_1$ connects to $t_1$ or $t_2$ for a random graph sampled from $\overline{\cD_{\blk}}$ with error probability at most $2/\sqrt{n}$ and sketch size $8L$, by Lemma~\ref{lem_blk_lb}, we must have $L\geq \Omega(\log^3 n)$.
This proves the theorem.
\end{proof}

\subsection{Sketch size lower bound for one block}
In this subsection, we prove Lemma~\ref{lem_blk_lb}.
We begin by defining a hard distribution $\cD_{\blk}$ that allows us to prove a lower bound on the expected \emph{maximum} sketch size.
Later, we will show how to extend it to expected average sketch size.

\paragraph{Hard distribution for one block $\cD_{\blk}$.}
For simplicity of notations, let us assume the vertices have labels from $-\frac12\sqrt{n}$ to $\frac12\sqrt{n}$.
We begin by describing the graph on positive labeled vertices, from $1$ to $\frac12\sqrt{n}$.
The main part consists of four sets $V^l, V^m, \tV^m, V^r$:
\begin{compactitem}
	\item $V^m$ and $\tV^m$ have $n^{1/4}$ vertices, and a perfect matching is placed between them;
	\item $V^r$ has $2n^{1/8}$ vertices, divided into two parts $V^r_1$ and $V^r_2$ of size $n^{1/8}$;
	\item $V^l$ consists of $n^{1/4}$ groups $V^l_1,\ldots,V^l_{n^{1/4}}$ of sizes \emph{at most} $n^{1/8}$.
\end{compactitem}
Thus, the four sets use in total at most $2n^{3/8}\ll \frac12\sqrt{n}$ vertices.
Each vertex $v^m_j\in V^m$ is associated with group $V^l_j\subset V^l$.
The only possible edges between the four sets are the matching between $V^m$ and $\tV^m$, the edges between $v^m_j$ and the associated $V^l_j$ and the edges between $V^m$ and $V^r$.

To construct such a graph, we first pick random $V^m,\tV^m,V^r_1$ and $V^r_2$ with the corresponding sizes, and place a uniformly random perfect matching between $V^m$ and $\tV^m$.
For each vertex $v^m_j\in V^m$, we independently sample a random instance $(S_j, T_j, P_{j,1}, P_{j,2})$ from the hard distribution $\cD_{\urdec}$ for $\urdec$ for $U=n^{1/8}$ and $\delta=4n^{-1/32}$, where $\cD_{\urdec}$ is the distribution in Lemma~\ref{lem_ur_lb_decision}.
Then we connect $v^m_j$ to $|T_j|$ random \emph{unused} vertices, and they form the set $V^l_j$.
If $S_j\setminus T_j\subseteq P_{j,1}$, we connect $v^m_j$ to $|S_j\setminus T_j|$ random vertices in $V^r_1$, otherwise, we connect it to $|S_j\setminus T_j|$ random vertices in $V^r_2$.
This completes the graph on positive-labeled vertices.

Next, we copy the subgraph to the vertices with negative labels.
That is, if vertices with labels $a,b>0$ have an edge between them, then we add an edge between vertices with labels $-a$ and $-b$.
Then we define the vertex sets ${}^-V^l, {}^-V^m, {}^-\tV^m, {}^-V^r$ over the negative labeled vertices similarly.

Finally, we connect the subgraph to the four special vertices $s_1,s_2,t_1,t_2$.
We connect all vertices in $V^r_1$ and ${}^-V^r_2$ to $t_1$, and all vertices in $V^r_2$ and ${}^-V^r_1$ to $t_2$.
We pick a random vertex $\tv^m_{j^*}\in \tV^m$ and connect it to $s_1$, then we connect ${}^-\tv^m_{j^*}$ to $s_2$.
At last, we connect all unused vertices to $t_1$.
See Figure~\ref{fig_hard_blk}.

It is not hard to verify that the block has two connected components, and $t_1$ and $t_2$ must be in different components.
Moreover, if $s_1$ is in the same component with $t_1$, then the path between them must go through $V^r_1$, in which case, there is a path from $s_2$ to $t_2$ going through ${}^-V^r_1$, i.e., $s_2$ and $t_2$ are in the same component.
Likewise, if $s_1$ is in the same component with $t_2$, then the path must go through $V^r_2$, and hence, $s_2$ and $t_1$  are in the same component.

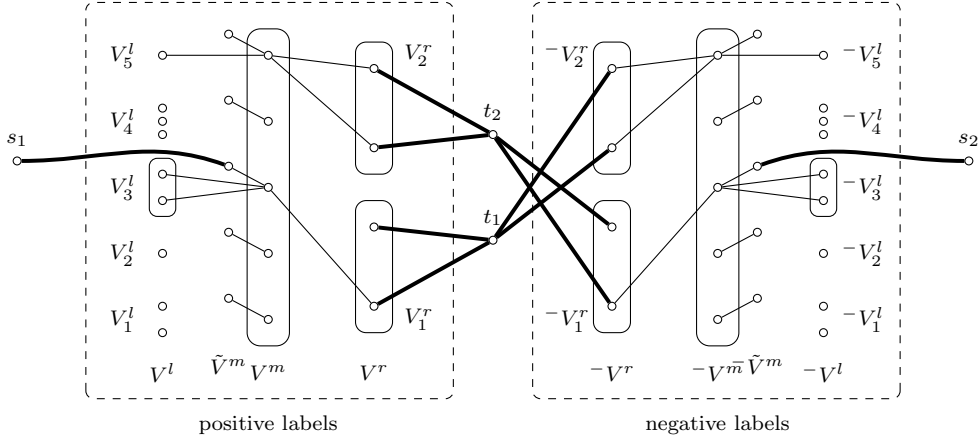
\begin{figure}
	\begin{center}
		\begin{tikzpicture}[vtx/.style={draw, circle, inner sep=0, minimum size=3pt}]
			\foreach \c in {-1,1} {
				\foreach \i in {0,1,2,3,4} {
					\node[vtx] (A\c\i) at (-\c*85+\c*0pt, \i*25 pt) {};
					\node[vtx] (tA\c\i) at (-\c*85-\c*15pt, \i*25+8 pt) {};
					\draw (A\c\i) -- (tA\c\i);
				}
				\foreach \i in {3}
					\foreach \j in {0,1,2}
						\node[vtx] (B\c\i\j) at (-\c*85-\c*40pt, \i*25+\j*5-5 pt) {};
				\foreach \i in {1,4}
					\foreach \j in {1}
						\node[vtx] (B\c\i\j) at (-\c*85-\c*40pt, \i*25+\j*5-5 pt) {};
				\foreach \i in {0,2}
					\foreach \j in {0,2}
						\node[vtx] (B\c\i\j) at (-\c*85-\c*40pt, \i*25+\j*5-5 pt) {};
				\foreach \i in {0,1}
					\foreach \j in {0,1}
						\node[vtx] (C\c\i\j) at (-\c*85+\c*40pt, \i*60+\j*30+5 pt) {};
				\draw [rounded corners=5pt] (-\c*85-\c*8pt, -10pt) rectangle (-\c*85+\c*8pt, 110pt);
				\draw [rounded corners=3pt] (-\c*85-\c*45pt, 39pt) rectangle (-\c*85-\c*35pt, 61pt);
				\draw [rounded corners=4pt] (-\c*85+\c*33pt, -5pt) rectangle (-\c*85+\c*47pt, 45pt);
				\draw [rounded corners=4pt] (-\c*85+\c*33pt, 55pt) rectangle (-\c*85+\c*47pt, 105pt);

				\node at (-\c*85-\c*40pt, -20pt) {\scriptsize \ifthenelse{\c=1}{$V^l$}{${}^-V^l$}};
				\node at (-\c*85+\c*0 pt, -20pt) {\scriptsize \ifthenelse{\c=1}{$V^m$}{${}^-V^m$}};
				\node at (-\c*85-\c*15 pt, -17pt) {\scriptsize \ifthenelse{\c=1}{$\tV^m$}{${}^-\tV^m$}};
				\node at (-\c*85+\c*40pt, -20pt) {\scriptsize \ifthenelse{\c=1}{$V^r$}{${}^-V^r$}};
				\foreach \i in {1,2,3,4,5}
					\node at (-\c*85-\c*55pt, \i*25-25 pt) {\scriptsize \ifthenelse{\c=1}{$V^l_\i$}{${}^-V^l_\i$}};
				\node at (-\c*85+\c*57pt, 0pt) {\scriptsize \ifthenelse{\c=1}{$V^r_1$}{${}^-V^r_1$}};
				\node at (-\c*85+\c*57pt, 100pt) {\scriptsize \ifthenelse{\c=1}{$V^r_2$}{${}^-V^r_2$}};

				\draw (B\c41) -- (A\c4) -- (C\c10);
				\draw (A\c4) -- (C\c11);

				\draw (B\c22) -- (A\c2) -- (C\c00);
				\draw (B\c20) -- (A\c2);
			}
			\draw [dashed, rounded corners=3pt] (-154pt, -30pt) rectangle (-15pt, 120pt);
			\draw [dashed, rounded corners=3pt] (154pt, -30pt) rectangle (15pt, 120pt);
			\node at (-85pt, -40pt) {\scriptsize positive labels};
			\node at (85pt, -40pt) {\scriptsize negative labels};

			\node [vtx, label={\scriptsize $t_2$}] (t2) at (0pt, 70pt) {};
			\node [vtx, label={\scriptsize $t_1$}] (t1) at (0pt, 30pt) {};
			\draw [line width=1.5pt] (C-110) -- (t1);
			\draw [line width=1.5pt] (C-111) -- (t1);
			\draw [line width=1.5pt] (C100) -- (t1);
			\draw [line width=1.5pt] (C101) -- (t1);
			\draw [line width=1.5pt] (C-100) -- (t2);
			\draw [line width=1.5pt] (C-101) -- (t2);
			\draw [line width=1.5pt] (C110) -- (t2);
			\draw [line width=1.5pt] (C111) -- (t2);

			\node [vtx, label={\scriptsize $s_2$}] (s2) at (180pt, 60pt) {};
			\node [vtx, label={\scriptsize $s_1$}] (s1) at (-180pt, 60pt) {};
			\draw (s2) edge[line width=1.5pt, out=180, in=20] (tA-12);
			\draw (s1) edge[line width=1.5pt, out=0, in=160] (tA12);
		\end{tikzpicture}
	\end{center}
	\caption{A hard instance for one block.}\label{fig_hard_blk}
\end{figure}

\bigskip

To decide whether $s_1$ is in the same component with $t_1$ or $t_2$, we need to solve the $\urdec$ instance embedded at the vertex $v^m_{j^*}$, which shares a common neighbor ($\tv^m_{j^*}$) with $s_1$.
This is because the neighbors of $v^m_{j^*}$ that are not in $V^l$ are all contained in either $V^r_1$ or $V^r_2$, and we need to decide which case it is (see below for more details).
Recall we have proved in the previous section that with error probability $\delta$, $\urdec$ requires message length at least $\Omega(\log(1/\delta)\log^2 U)$, which is $\Omega(\log^3 n)$ for our setting of parameters.
We restate the lower bound below.
\begin{restate}[Lemma~\ref{lem_ur_lb_decision}]
	\lemurlbdecisioncont
\end{restate}

To prove Lemma~\ref{lem_blk_lb}, we apply an embedding argument similar to \cite{NY19} to make a reduction from $\urdec$, and then apply the $\urdec$ lower bound.
Given an $\urdec$ instance $(S, T, P_1, P_2)$ for $U=n^{1/8}$, consider the following procedure to construct a graph $G_{\blk}$ for a block on vertices labeled from $-\frac12\sqrt{n}$ to $\frac12\sqrt{n}$ (note that this procedure as is may not be completed by either player without communication):
\begin{compactenum}
	\item pick random $V^m,\tV^m$ of size $n^{1/4}$ from the vertices with positive labels, and place a uniformly random perfect matching between them;
	\item pick a random vertex $v^m_{j^*}\in V^m$, let $\tv^m_{j^*}\in\tV^m$ be the vertex it matches to;
	\item pick a random \emph{injection} $f:[U]\rightarrow \{1,\ldots,\frac12\sqrt{n}\}\setminus (V^m\cup \tV^m)$;
	\item set $V^l_{j^*}$ to $f(T)$;
	\item set $V^r_1$ to the union of $f(P_1)$ and $n^{1/8}-|P_1|$ other random unused vertices;
	\item set $V^r_2$ to the union of $f(P_2)$ and $n^{1/8}-|P_2|$ other random unused vertices;
	\item connect $v^m_{j^*}$ to $f(S)$;
	\item sample the neighborhoods of $V^m\setminus \{v^m_{j^*}\}$ according to $\cD_{\blk}$;
	\item copy the graph to negative labeled vertices according to $\cD_{\blk}$;
	\item connect $s_1$ to $\tv^m_{j^*}$, $s_2$ to ${}^-\tv^m_{j^*}$, $t_1$ to all vertices in $V^r_1$ and ${}^-V^r_2$, $t_2$ to all vertices in $V^r_2$ and ${}^-V^r_1$.
\end{compactenum}
Note that $v^m_{j^*}$ connects to all $|T|$ vertices in $V^l_{j^*}$, it connects to $|S\setminus T|$ vertices in $V^r$, which are all in either $V^r_1$ or $V^r_2$.
When $(S, T, P_1, P_2)$ is sampled from $\cD_{\urdec}$, the neighborhood of $v^m_{j^*}$ follows $\cD_{\blk}$.
Since the rest of the graph is also sampled according to $\cD_{\blk}$, the whole graph follows the hard distribution $\cD_{\blk}$.
Moreover, $S\setminus T\subset P_1$ if $s_1$ and $t_1$ are in the same connected component, and $S\setminus T\subset P_2$ if $s_1$ and $t_2$ are in the same component.

Denote by $\mu$ the joint distribution of $S,T,P_1,P_2$ and $G_{\blk}$ following the above procedure.
We use ${}^\pm V^m$ to denote $V^m\cup {}^-V^m$, and ${}^\pm V^l,{}^\pm \tV^m,{}^\pm V^r$ are defined similarly.
For a vertex $v$, we denote the sketch of its neighborhood by $\sk(v)$.
Similarly for a set of vertices $V$, $\sk(V)$ denotes the collection of all its sketches.

Suppose there is a protocol $\cP_{\blk}$ that decides if $s_1$ is the same component with $t_1$ or $t_2$ with error probability $2/\sqrt{n}$ such that
\begin{compactitem}
	\item the expected \emph{average} sketch size of ${}^\pm V^m$ is at most $L$, and
	\item the expected \emph{average} sketch size of ${}^\pm V^r$ is at most $L$.
\end{compactitem}
Note that both conditions are implied if the expected \emph{maximum} sketch size is at most $L$.
We are going to use this protocol to solve the communication problem $\urdec$ using $O(L)$ bits of communication in expectation and with low error probability.

\paragraph{Protocol for $\urdec$.}
The players first sample $V^m,\tV^m,v^m_{j^*},f$ and the perfect matching $\Pi$ using public random bits according to $\mu$ (step 1 to step 3).
Then Alice, who knows $S$, privately computes $f(S)$ (step 7), which together with $\tv^m_{j^*}$ is the neighborhood of $v^m_{j^*}$, then she simulates $\cP_{\blk}$ as $v^m_{j^*}$ and its copy ${}^-v^m_{j^*}$, and sends the sketches $\sk(v^m_{j^*})$ and $\sk({}^-v^m_{j^*})$ to Bob.
Bob, who knows $T,P_1,P_2$, computes $f(T),f(P_1),f(P_2)$, and samples $V^r_1,V^r_2$ and $V^l_{j^*}$ according to $\mu$ (step 4 to step 6).
Then he samples the neighborhood for all vertices in $V^m\setminus \{v^m_{j^*}\}$ according to $\mu$ (step 8).
Now, Bob knows the sets $V^l_1,\ldots,V^l_{n^{1/4}},V^m,\tV^m,V^r_1,V^r_2$, and he knows the neighborhoods of $s_1,s_2,t_1,t_2$ and the neighborhoods of all vertices in $V^l$, $V^m\setminus \{v^m_{j^*}\}$, $\tV^m$.
Bob computes the sketches for all these vertices and the sketches for their copies with negative labels.
Together with Alice's message, Bob knows $\sk({}^\pm V^l),\sk({}^\pm V^m),\sk({}^\pm \tV^m),\sk(s_1),\sk(s_2),\sk(t_1),\sk(t_2)$.
Bob examines the posterior distribution of the neighborhood of $v^m_{j^*}$ \emph{conditioned on} 
\begin{compactitem}
	\item the sets ${}^\pm V^l$, ${}^\pm V^m$, ${}^\pm \tV^m$, ${}^\pm V^r$, 
	\item the matching $\Pi$ between $V^m$ and $\tV^m$,
	\item the index $j^*$, and
	\item the sketches $\sk({}^\pm V^l),\sk({}^\pm V^m),\sk({}^\pm \tV^m),\sk(s_1),\sk(s_2),\sk(t_1),\sk(t_2)$.
\end{compactitem}
If in this posterior distribution, $v^m_{j^*}$ connects to $V^r_1$ with probability at least $1/2$, Bob returns ``$S\setminus T\subseteq P_1$'', otherwise, he returns ``$S\setminus T\subseteq P_2$''.

\paragraph{Communication cost.} The only message in the above protocol is the two sketches $\sk(v^m_{j^*})$ and $\sk({}^-v^m_{j^*})$. 
Since $v^m_{j^*}$ is a random vertex in $V^m$, the expected length $|\sk(v^m_{j^*})|$ is simply the expected \emph{average} sketch size of vertices in $V^m$.
Similarly, the expected length $|\sk({}^-v^m_{j^*})|$ is the expected average sketch size of ${}^-V^m$.
By the assumption of $\cP_{\blk}$, the expected message length is at most $2L$.

\paragraph{Error probability.} It remains to analyze the error probability of the protocol.
If at the end of the protocol, Bob also knew $\sk({}^\pm V^r)$, then by simulating $\cP_{\blk}$ as the referee, Bob would be able to detect if $s_1$ is in the same component with $t_1$ or $t_2$, with an overall error probability of at most $2/\sqrt n$ on a random instance.
In particular, he would be able to decide if $v^m_{j^*}$ has its neighbors in $V^r_1$ or $V^r_2$, i.e., $S\setminus T\subset P_1$ or $S\setminus T\subset P_2$.
In the other words, in the posterior distribution of the neighborhood of $v^m_{j^*}$ as in the protocol but further conditioned on $\sk({}^\pm V^r)$, let $\epsilon$ be such that $v^m_{j^*}$ has no neighbors in $V^r_1$ with probability $1-\epsilon$, then we must have $\E[\min\{\epsilon,1-\epsilon\}]$ upper bounded by the overall error probability $2/\sqrt n$.
To upper bound the error probability of the protocol, we are going to show that whether we condition on $\sk({}^\pm V^r)$ does not distort the posterior distribution by much in expectation.

The expected total size of $\sk({}^\pm V^r)$ is at most $2Ln^{1/8}$ by the assumption of $\cP_{\blk}$.
Denote by $N(v)$ the neighborhood of $v$.
We have the mutual information
\[
	I(\sk({}^\pm V^r);N(v^m_1),\ldots,N(v^m_{n^{1/4}})\mid \Pi, \sk({}^\pm V^l),\sk({}^\pm V^m),\sk({}^\pm\tV^m),\sk(s_1),\sk(s_2),\sk(t_1),\sk(t_2))\leq 2Ln^{1/8},
\]
where $\Pi$ is the matching between $V^m,\tV^m$, and for simplicity of notations, we omitted the sets $V^l,V^m,\tV^m,V^r$ in the condition.
Then observe that conditioned on $\Pi, \sk({}^\pm V^l),\sk({}^\pm V^m)$, we have $\sk({}^\pm V^r)$ and $N(v^m_1),\ldots,N(v^m_{n^{1/4}})$ are independent of $\sk({}^\pm\tV^m)$, $\sk(s_1)$, $\sk(s_2)$, $\sk(t_1)$, $\sk(t_2)$.
To see this,
\begin{compactitem}
 	\item the neighborhoods of $t_1$ and $t_2$ are deterministic given the sets $V^r_1, V^r_2$;
 	\item each vertex in $\tV^m$ has a fixed neighbor in $V^m$ given the matching;
 	\item one vertex in $\tV^m$ [resp. ${}^-\tV^m$] has $s_1$ [resp. $s_2$] as its neighbor, which is determined independent of the rest of the graph.
\end{compactitem}
Hence, we may remove them from the condition,
\[
	I(\sk({}^\pm V^r);N(v^m_1),\ldots,N(v^m_{n^{1/4}})\mid \Pi, \sk({}^\pm V^l),\sk({}^\pm V^m))\leq 2Ln^{1/8}.
\]
Next, observe that $N(v^m_1),\ldots,N(v^m_{n^{1/4}})$ are still independent even conditioned on $\Pi, \sk({}^\pm V^l),\sk({}^\pm V^m)$.
By the superadditivity of mutual information with independent random variables, we have
\[
	\sum_{j=1}^{n^{1/4}} I(\sk({}^\pm V^r);N(v^m_j)\mid \Pi, \sk({}^\pm V^l),\sk({}^\pm V^m))\leq 2Ln^{1/8}.
\]
Since conditioned on $\Pi, \sk({}^\pm V^l),\sk({}^\pm V^m),\sk({}^\pm V^r),N(v^m_1),\ldots,N(v^m_{n^{1/4}})$, $j^*$ is still uniformly random, we have
\begin{align*}
	&\phantom{=}\,\,I(\sk({}^\pm V^r);N(v^m_{j^*})\mid j^*, \Pi, \sk({}^\pm V^l),\sk({}^\pm V^m)) \\
	&=\sum_{j=1}^{n^{1/4}}\frac{1}{n^{1/4}}\cdot I(\sk({}^\pm V^r);N(v^m_{j^*})\mid j^*=j,\Pi, \sk({}^\pm V^l),\sk({}^\pm V^m)) \\
	&=\frac{1}{n^{1/4}}\cdot\sum_{j=1}^{n^{1/4}} I(\sk({}^\pm V^r);N(v^m_{j})\mid \Pi, \sk({}^\pm V^l),\sk({}^\pm V^m)) \\
	&\leq 2Ln^{-1/8}.
\end{align*}

Let $\dist_{\mu}(X\mid Y)$ denote the distribution of $X$ conditioned on $Y$.
By Pinsker's inequality, concavity of square root and the fact that mutual information is equal to the expected KL-divergence, we have
\begin{align*}
	&\phantom{=}\,\,\E[\|\dist_{\mu}(N(v^m_{j^*})\mid j^*, \Pi, \sk({}^\pm V^l),\sk({}^\pm V^m))-\dist_{\mu}(N(v^m_{j^*})\mid j^*, \Pi, \sk({}^\pm V^l),\sk({}^\pm V^m),\sk({}^\pm V^r))\|_1] \\
	&\leq \E\left[\sqrt{2\DKLver{N(v^m_{j^*})\mid j^*, \Pi, \sk({}^\pm V^l),\sk({}^\pm V^m)}{N(v^m_{j^*})\mid j^*, \Pi, \sk({}^\pm V^l),\sk({}^\pm V^m),\sk({}^\pm V^r)}}\right] \\
	&\leq \sqrt{\E\left[{2\DKLver{N(v^m_{j^*})\mid j^*, \Pi, \sk({}^\pm V^l),\sk({}^\pm V^m)}{N(v^m_{j^*})\mid j^*, \Pi, \sk({}^\pm V^l),\sk({}^\pm V^m),\sk({}^\pm V^r)}}\right]} \\
	&=\sqrt{2I(\sk({}^\pm V^r);N(v^m_{j^*})\mid j^*, \Pi, \sk({}^\pm V^l),\sk({}^\pm V^m))} \\
	&\leq \sqrt{4Ln^{-1/8}}.
\end{align*}

Again by the fact that $N(v^m_{j^*})$ is independent of $\sk({}^\pm \tV^m)$ and $\sk(s_1),\sk(s_2),\sk(t_1),\sk(t_2)$, conditioned on $j^*, \Pi, \sk({}^\pm V^l),\sk({}^\pm V^m)$, or conditioned on $j^*, \Pi, \sk({}^\pm V^l),\sk({}^\pm V^m),\sk({}^\pm V^r)$, the distribution 
\[
	\dist_{\mu}(N(v^m_{j^*})\mid j^*, \Pi, \sk(\pm V^l),\sk(\pm V^m),\sk({}^\pm \tV^m),\sk(s_1),\sk(s_2),\sk(t_1),\sk(t_2))
\]
is $\sqrt{4Ln^{-1/8}}$-close to 
\[
	\dist_{\mu}(N(v^m_{j^*})\mid j^*, \Pi, \sk(\pm V^l),\sk(\pm V^m),\sk(\pm \tV^m),\sk(\pm V^r),\sk(s_1),\sk(s_2),\sk(t_1),\sk(t_2))
\]
in expectation.
Note that the former distribution is exactly what Bob examines.
However, we know that in the latter distribution, $N(v^m_{j^*})$ is disjoint from $V^r_1$ with probability $1-\epsilon$ such that $\E[\min\{\epsilon,1-\epsilon\}]\leq 2/\sqrt n$.
Hence, in the former distribution, we also have $\E[\min\{\epsilon,1-\epsilon\}]\leq 2/\sqrt n+\sqrt{4Ln^{-1/8}}$, which is at most $4n^{-1/32}$ when $L\leq n^{1/16}$.
By answering $S\setminus T\subset P_1$ if $\epsilon>1/2$ and $S\setminus T\subset P_2$ if $\epsilon\leq 1/2$, the overall error probability of the protocol is at most $\delta=4n^{-1/32}$.
Finally, by Lemma~\ref{lem_ur_lb_decision}, we must have $L\geq \min\{n^{1/16},\Omega(\log^3 n)\}=\Omega(\log^3 n)$.

\subsection{Extending to average sketch size} The above argument shows that if the error probability of the sketching scheme for a block is at most $2/\sqrt n$, and the expected average sketch size of ${}^\pm V^m$ and that of ${}^\pm V^r$ are both at most $L$, then $L$ must be at least $\Omega(\log^3 n)$.
However, since $|V^m|+|V^r|\ll \sqrt n$, it does not directly prove a lower bound on the expected average sketch size of \emph{all} vertices.
In the following, we show how to prove the same lower bound on $L$ when the expected average sketch size of all vertices is at most $L$.
The main idea is simple: with constant probability, we construct a graph such that most vertices have neighborhoods that look like those of $V^m$; with constant probability, most vertices have neighborhoods that look like those of $V^r$.
Therefore, if the overall average sketch size is $L$, then it implies that the expected average sketch sizes of $V^m$ and $V^r$ are both at most $O(L)$.
We also need to ensure that the block always consists of two connected components such that $s_1,s_2$ are in different components and $t_1,t_2$ are in different components.
We begin by describing the hard distribution $\overline{\cD_{\blk}}$.

\paragraph{Hard distribution $\overline{\cD_{\blk}}$.} 
Let $\cD_{\textrm{deg}, m}$ be the degree distribution of a vertex in $V^m$ according to $\cD_{\blk}$.
Then for every $v\in V^m$, the marginal distribution of its neighborhood is $d$ uniformly random vertices, for $d$ following $\cD_{\textrm{deg}, m}$.
Similarly, let $\cD_{\textrm{deg}, r}$ be the degree distribution of a vertex in $V^r$.
Then for $v\in V^r_1$ [resp. $v\in V^r_2$], the marginal distribution of its neighborhood is $d$ uniformly random vertices, for $d$ following $\cD_{\textrm{deg}, r}$, \emph{conditioned on} $t_1$ [resp. $t_2$] being its neighbor.
In the distribution $\overline{\cD_{\blk}}$, we randomly choose one of the following three procedures to generate the block.
\begin{enumerate}[(i)]
	\item We sample the block from the previous distribution $\cD_{\blk}$.
	\item We choose between the following two cases randomly: connect $s_1$ to $t_1$ and $s_2$ to $t_2$; connect $s_1$ to $t_2$ and $s_2$ to $t_1$.
	We pick half of the vertices $S$ with positive labels, and let $\overline{S}$ be the remaining half.
	For each vertex $v\in S$, we sample its degree $d_v$ according to $\cD_{\textrm{deg}, m}$, and sample $d_v$ vertices in $\overline{S}$ to be its neighbors.
	Then we connect all $\overline{S}$ to $t_1$.
	Finally, we copy the graph (as well as the incident edges to $t_1$) to the negative-labeled vertices.
	\item We choose between the following two cases randomly: connect $s_1$ to $t_1$ and $s_2$ to $t_2$; connect $s_1$ to $t_2$ and $s_2$ to $t_1$.
	We partition the remaining positive labeled vertices into four sets of equal sizes $S_1,\overline{S}_1,S_2,\overline{S}_2$.
	For each vertex $v\in S_1$ [resp. $v\in S_2$], we sample its degree $d_v$ according to $\cD_{\textrm{deg}, r}$, sample $d_v-1$ vertices in $\overline{S}_1$ [resp. $v\in \overline{S}_2$] to be its neighbors and connect $v$ to $t_1$ [resp. $t_2$].
	Then we connect all $\overline{S}_1$ to $t_1$ and $\overline{S}_2$ to $t_2$.
	Finally we copy the graph (as well as the incident edges to $t_1$ and $t_2$) to the negative-labeled vertices.
\end{enumerate}
It is easy to verify that the block always has two connected components such that $s_1,s_2$ are in different components and $t_1,t_2$ are in different components.

If there is a protocol that solves an instance sampled from $\overline{\cD_{\blk}}$ with error probability $2/\sqrt n$ and expected average sketch size $L$.
Then its error probability conditioned on choosing procedure (i) is at most $6/\sqrt{n}$, i.e., the error probability for $\cD_{\blk}$ is at most $6/\sqrt{n}$.
Moreover, its expected average sketch size conditioned on choosing procedure (ii) is at most $3L$.
Since a constant fraction of the vertices in this case have their neighborhoods identically distributed as vertices in ${}^\pm V^m$ according to $\cD_{\blk}$.
It implies that the expected average sketch size of ${}^\pm V^m$ on a instance sampled from $\cD_{\blk}$ is at most $O(L)$.
Similarly, from procedure (iii), we obtain that the expected average sketch size of ${}^\pm V^r$ on a instance sampled from $\cD_{\blk}$ is also at most $O(L)$.
Finally, by the argument from the previous subsection, we conclude that $L\geq \Omega(\log^3n)$.
This proves Lemma~\ref{lem_blk_lb}.
\bibliography{refs}
\bibliographystyle{alpha}

\end{document}